\documentclass[sigconf,nonacm]{aamas} 

\usepackage{algorithm}
\usepackage{algorithmic}

\usepackage{balance}
\usepackage{enumitem}

\doi{HTED8429}

\setcopyright{ifaamas}
\acmConference[AAMAS '26]{Proc.\@ of the 25th International Conference
on Autonomous Agents and Multiagent Systems (AAMAS 2026)}{May 25 -- 29, 2026}
{Paphos, Cyprus}{C.~Amato, L.~Dennis, V.~Mascardi, J.~Thangarajah (eds.)}
\copyrightyear{2026}
\acmYear{2026}
\acmDOI{}
\acmPrice{}
\acmISBN{}

\acmSubmissionID{<<OpenReview submission id>>}

\title[AAMAS-2026 Formatting Instructions]{Interbank Lending Games}

\author{Jinyun Tong}
\affiliation{
  \institution{King's College London}
  \city{London}
  \country{United Kingdom}}
\email{jinyun.tong@kcl.ac.uk}

\author{Bart de Keijzer}
\affiliation{
  \institution{King's College London}
  \city{London}
  \country{United Kingdom}}
\email{bart.de\_keijzer@kcl.ac.uk}

\author{Haoxiang Wang}
\affiliation{
  \institution{King's College London}
  \city{London}
  \country{United Kingdom}}
\email{whx@kcl.ac.uk}

\author{Carmine Ventre}
\affiliation{
  \institution{King's College London}
  \city{London}
  \country{United Kingdom}}
\email{carmine.ventre@kcl.ac.uk}

\begin{abstract}
    We define and study a lending game to model the interbank money market, in which lending banks strategically allocate their cash to borrowing banks. The interest rate offered by each borrowing bank is within the interest rate corridor set by the central bank and ultimately depends on the demand and the supply of cash in the interbank market. Lending banks naturally aim to maximise the income coming from the interest repayments. In its purest form, this is an infinite-strategy game that we show to be an exact potential game which has a unique pure strategy Nash equilibrium. We then define and solve a constrained optimisation problem and propose a strongly polynomial-time algorithm to compute this Nash equilibrium. We also study some variants of best-response dynamics of this lending game, showing that they converge to the Nash equilibrium in both discrete and continuous-time scenarios.
\end{abstract}


\newcommand{\BibTeX}{\rm B\kern-.05em{\sc i\kern-.025em b}\kern-.08em\TeX}

\begin{document}


\pagestyle{fancy}
\fancyhead{}


\maketitle 


\section{Introduction}
Bank operations are constrained by their funding liquidity. Banks' cash inflows and outflows may not be aligned in time, exposing such banks to default risks, as they cannot repay their due debts, even though their balance sheets may show that they are solvent. To avoid liquidity defaults, banks may borrow cash reserves from other banks. The ensuing interbank money market serves as a place where banks can borrow from or lend to one another on a short-term basis. The interbank money market is a typical over-the-counter market where banks conduct transactions through direct contact or broker matching. Such an over-the-counter market allows borrowing and lending banks the flexibility to customise their loan terms, such as lending amount, maturity and interest rate.

Interest rates in the interbank money market are usually influenced by the central bank through interest rate policies and open market operations. For example, the European Central Bank (ECB) uses the interest rate corridor policy to effectively control the market interest rate fluctuations of the interbank money market within the target range \cite{brandao2024ecb}. Specifically, the ECB sets two interest rates, the \textit{deposit facility rate} and the \textit{marginal lending facility rate} \cite{ECB2023key}. When commercial banks have excess liquidity, they can choose to deposit the excess liquidity in the overnight deposit facility of the ECB, and the ECB will provide them with a relatively low deposit facility rate. On the other hand, if a bank cannot borrow sufficient liquidity through the interbank money market to fill its liquidity gap, it must apply for an overnight loan from the ECB to avoid default; the loan will then come with a relatively high marginal lending facility rate. Thus, the interest rate in the interbank money market will not be lower than the deposit facility rate, since otherwise, banks will deposit liquidity into the central bank instead of lending it out. Moreover, the market interest rate will not be higher than the lending facility rate, for otherwise, banks would rather apply for overnight loans from the central bank than borrow in the interbank money market. Therefore, these two important interest rates set by the central bank are usually considered the boundaries of the interest rate corridor, and market interest rates fluctuate between them. Naturally, interest rates in the interbank money market are also affected by supply and demand. In general, the higher the supply of funds in the market, the lower the market interest rate, and vice versa. In this work, we study the interest rate corridor policy, where the individual interest rate is determined by the interest rate corridor and the ratio between the supply available to the borrowing banks and their demand.

In this context, we focus on bank incentives and strategic decisions in the interbank money market, especially the lending amounts and interest rates at equilibrium under certain market conditions. We achieve this goal by employing a game-theoretical analysis of a strategic game with an infinite strategy space. Our work may thus contribute to understanding the incentive of interbank lending and to making the interest rate corridor policy. We model money as a continuous quantity, in that the lending amounts can be any non-negative real number within the given constraints. This allows us to capture the nature of the incentives and dynamics by avoiding the distortions due to the chosen denomination of one unit of cash and the granularity of exchange rates.

\subsection{Our Contributions}
We introduce \textit{interbank lending games}: a class of non-cooperative games, which we define to study bank behaviour in the interbank money market. In Section \ref{section:preliminaries}, we formally define our model of the interbank money market and the lending game, following the aforementioned fundamental market rules. Then, we define a potential function and show that the lending game is an exact potential game with the proposed potential function. In Section \ref{section:equilibrium}, we first prove that this lending game has a unique pure Nash equilibrium by showing that the potential function is a strictly concave function that has a global maximum over the strategy space. Then, we give an explicit formulation of the unique pure Nash equilibrium by solving a constrained optimisation problem. In addition, since this solution is not trivial to compute, we design an algorithm to iteratively compute the pure Nash equilibrium in strongly polynomial time.

A Nash equilibrium is not necessarily the ``final'' strategy profile that the agents attain in a multi-agent system, considering that players may never reach it from a given initial strategy profile. Therefore, we are also interested in whether the lending banks in our model are guaranteed to converge to the Nash equilibrium through natural iterative dynamics. In Section \ref{section:dynamics}, we define and analyse some variants of the best-response dynamics of the lending game in both discrete-time and continuous-time scenarios. We show that in the discrete-time scenario, the \textit{eager best-response dynamics} and the \textit{randomised best-response dynamics} converge to the pure Nash equilibrium under asynchronous updates (i.e., where in each iteration a single agent best-responds). For synchronous updates (i.e., where all agents simultaneously best-respond to the given strategy profile), we apply \textit{pseudo-gradient dynamics} and prove the convergence of the learning agents to the pure Nash equilibrium. Finally, we also show the convergence of continuous-time best-response dynamics by defining a Lyapunov function.

\subsection{Related Work}
Monderer and Shapley \cite{monderer1996potential} pioneered some definitions of various \textit{potential games}, in which a global potential function captures the incentive of all players to change their strategies. They showed that any \textit{finite potential game} (i.e., potential game with a finite number of players and finite strategy sets) must have at least one pure Nash equilibrium. From the perspective of the improvement path, the existence of equilibrium can be explained by the finite improvement property as the game has a finite number of strategy profiles and every improvement path is also finite \cite{monderer1996potential}. Nevertheless, for \textit{infinite potential games} (i.e., potential games with an infinite number of players or infinite strategy sets), this property is not guaranteed \cite{voorneveld1997equilibria}. For games with an infinite number of players, the concepts of \textit{large population games} \cite{huang2006large} and \textit{mean-field games} \cite{lasry2007mean} were developed. On the other hand, if each player has a continuum as their strategy set (as is the case in our lending games), the game will have an infinite strategy set. In this case, even if it has a potential function, the existence of equilibria is not a given and relies on additional properties of the strategy space and potential function.

Monderer and Shapley \cite{monderer1996potential} were the first to explain the relationship between Nash equilibria of potential games and the extreme values of the potential functions. They proved that each Nash equilibrium of the potential game corresponds to a local maximiser of the potential function, and vice versa. This conclusion provides a new optimisation perspective for equilibrium analysis and has inspired a series of subsequent studies. Neyman \cite{neyman1997correlated} showed that for a potential game with a compact and convex strategy set, if its potential function is smooth and concave, correlated equilibria of the game are mixtures of pure Nash equilibria that maximise the potential function. This result was later strengthened in various ways, by Ui \cite{ui2008correlated}, Christensen \cite{christensen2017necessary}, Cai et al. \cite{cai2019role}, and Cao et al. \cite{cao2025correlated}. Based on this literature, an effective method of finding and characterising the pure Nash equilibria of a potential game is to find local maxima of the potential function. This translates the problem of finding an equilibrium to solving an optimisation problem. In particular, if the objective function of this associated optimisation problem is convex (or concave, for maximisation), it can in important cases be solved in weakly polynomial time through convex quadratic programming \cite{kozlov1980polynomial,ye1989extension}. It turns out that this approach indeed also applies to our class of lending games.

In terms of the convergence of the best-response dynamics, Hofbauer and Sorin \cite{hofbauer2006best} studied continuous-time best-response dynamics in continuous concave-convex zero-sum games, proving that its trajectory converges to the set of saddle points of the game. Barron et al. \cite{barron2010best} generalised this conclusion to differential best-response dynamics in continuous non-concave non-convex games; however, they found that for three-player non-zero-sum games, the convergence of the dynamics is not guaranteed. Leslie et al. \cite{leslie2020best} defined and analysed best-response dynamics in two-player zero-sum stochastic games and proved convergence to the set of Nash equilibrium strategies. Swenson et al. \cite{swenson2018best} showed that the best-response dynamics of almost all \text{regular potential games} which satisfy the regularity conditions introduced in \cite{harsanyi1973oddness} always converge to a pure Nash equilibrium, from almost every initial state. They also established that the convergence trajectory for almost every initial condition is unique, and analysed its convergence rate.

Another related line of research focuses on applying game-theoretic analysis to systemic risk in financial networks. Building on the agent-based model in \cite{eisenberg2001systemic}, a number of papers study both computational complexity and game-theoretic questions with the solvency risk of financial institutions. For example, the extent to which balance sheet transformations (e.g., forgiving debts) can lead to better outcomes is considered in \cite{kanellopoulos2022forgiving,kanellopoulos2023debt,tong2024reducing,tong2024selfishly,papp2020network}. The hardness of computing risk exposure in the presence of financial derivatives is considered in \cite{schuldenzucker2017finding,ioannidis2022financial,ioannidis2022strong}. The literature on liquidity risk is somewhat sparser with strategies to balance solvency and liquidity considered in \cite{cont2020liquidity} and analysed in a network context in \cite{zhao2023liquidity}.

\section{Preliminaries}\label{section:preliminaries}
In this section, we formally define the interbank money market and the lending game based on the fundamental market rules and mechanisms discussed in the introduction.

\subsection{Lending Games}
An \textit{interbank lending game} $G$ (or simply a \textit{lending game}) is given by the tuple $(m,n,\boldsymbol{c},\boldsymbol{d},r_{\text{max}},r_{\text{min}})$, where $\boldsymbol{c}=(c_{1},\dots,c_{m})\in\mathbb{R}_{>0}^{m}$ and $\boldsymbol{d}=(d_{1},\dots,d_{n})\in\mathbb{R}_{>0}^{n}$. For such a game, we refer to $L=\{1,\dots,m\}$ as the set of \textit{lenders} and we refer to $B=\{1,\ldots,n\}$ as the set of \textit{borrowers}. The set $L$ takes the role of the player set. The value $c_{i}$ represents the cash reserves available for lending of lender $i\in L$, which is called $i$'s \textit{budget}, and $d_{j}$ represents the borrowing demand of borrower $j\in B$. The values $r_{\text{min}}$ and $r_{\text{max}}$ represent the minimum and maximum interest rates set by the central bank. We proceed below to define the strategy sets and the utility functions of each of the lenders. In the subsequent sections, we use $G$ to refer to a lending game $(m,n,\boldsymbol{c},\boldsymbol{d},r_{\text{max}},r_{\text{min}})$.

\textbf{Strategies.} Each lender $i\in L$ proposes a lending vector $s_{i}=(s_{i1},\dots,s_{in})$, where $s_{ij}$ represents the budget allocated to the borrower $j\in B$, which can be any non-negative real number within the feasibility constraints in (\ref{equation:strategy_set}) below. This lending vector is called $i$'s \textit{strategy}. The \textit{strategy set} of each lender $i\in L$ is defined as
\begin{align}\label{equation:strategy_set}
    S_{i}=\left\{s_{i}\in\mathbb{R}_{\geq0}^{n} \ \middle\vert \ \sum_{j\in B} s_{ij}\leq c_{i}\right\}
\end{align}
which implies that the total lending amount of each lender cannot exceed its budget. The \textit{strategy space} of the game $G$, denoted by $\boldsymbol{S}$, is given by the joint strategy set of all lenders, i.e.,
\begin{align*}
    \boldsymbol{S}=\bigtimes_{i\in L}S_{i}.
\end{align*}
A \textit{strategy profile} defines a vector in which each lender selects its own strategy, corresponding to a group of \textit{lending amounts} (denoted by $s_{ij}$), i.e.,
\begin{align*}
    \boldsymbol{s}=(s_{i})_{i\in L}=(s_{ij})_{i\in L,j\in B}\in\boldsymbol{S}.
\end{align*}
For a given lender $i\in L$, let $\boldsymbol{s}_{-i}$ denote the strategy profile of all lenders except $i$, i.e., $\boldsymbol{s}_{-i}=(s_{k})_{k\in L\setminus\{i\}}\in\boldsymbol{S}_{-i}$, where the reduced strategy space is given by $\boldsymbol{S}_{-i}=\bigtimes_{k\in L\setminus\{i\}}S_{k}$. For a lender $i\in L$, strategy $s_{i}\in S_{i}$ and strategy profile $\boldsymbol{s}_{-i}\in\boldsymbol{S}_{-i}$, we use the standard notation $(s_{i},\boldsymbol{s}_{-i})$ to denote the strategy profile in $\boldsymbol{S}$ where lender $i$'s strategy is $s_{i}$ and the remaining lenders play according to $\boldsymbol{s}_{-i}$.

\textbf{Utility.} Each lender is assumed to strategise so as to maximise a utility function which is derived from the interest payments of funded borrowers. For a given strategy profile $\boldsymbol{s}\in\boldsymbol{S}$, the \textit{interest rate} offered by a borrower $j\in B$ is defined by the following linear function,
\begin{align}\label{equation:interest_rate}
    r_{j}(\boldsymbol{s})=(r_{\text{min}}-r_{\text{max}})\frac{\sum_{i\in L}s_{ij}}{d_{j}}+r_{\text{max}}
\end{align}
where $r_{\text{min}}$ and $r_{\text{max}}$ are the minimum and maximum interest rates determined by regulators, respectively, satisfying $0<r_{\text{min}}<r_{\text{max}}$. In the ECB scenario mentioned in the introduction, $r_{\text{min}}$ is the deposit facility rate and $r_{\text{max}}$ is the marginal lending facility rate. This definition indicates that the interest rate offered by the borrower $j$ varies between $r_{\text{min}}$ (when the supply $j$ receives is sufficient to cover its demand) and $r_{\text{max}}$ (when $j$ receives no supply), if it is not oversupplied (i.e., $\sum_{i\in L}s_{ij}\leq d_{j}$). If $j$ is oversupplied, the interest rate offered will be lower than $r_{\text{min}}$. For each lender $i\in L$, the \textit{utility function} is the mapping from a strategy profile $\boldsymbol{s}$ to $\mathbb{R}$ given by \begin{align}\label{equation:utility_function}
    u_{i}(\boldsymbol{s})=\sum_{j\in B}(r_{j}(\boldsymbol{s})-r_{\text{min}})s_{ij},
\end{align}
representing the sum of profits from interest payments received by $i$ from all borrowers in excess to those from the deposit facility of the central bank.\footnote{An alternative natural way to define the utility of a lender $i$ would be as the sum of interest received from the borrowers and central bank combined, i.e., $u_{i}(\boldsymbol{s})=\sum_{j\in B}r_{j}(\boldsymbol{s})s_{ij}+r_{\text{min}}(c_{i}-\sum_{j\in B}s_{ij})$. The latter yields the same strategic behaviour by the lenders as the latter definition is equivalent to adding the constant term $r_{\text{min}}c_{i}$ to $i$'s utility function.} Importantly, if the interest rate offered by a particular borrower $j\in B$ is lower than $r_{\text{min}}$ due to oversupply, lenders will receive a negative utility, and will prefer to deposit the oversupply in the central bank deposits for a return of $r_{\text{min}}$. Hence, lenders will always deviate from strategy profiles that oversupply any given borrower in the game, and in particular, such strategy profiles are never Nash equilibria.

\subsection{Pure Nash Equilibria and Potential Games}
Nash equilibria form a central solution concept in non-cooperative games: They are strategy profiles where no player can improve its utility by unilaterally changing its strategy, given that the other players remain playing their current strategies. Formally, a strategy profile $\boldsymbol{s}^{\ast}=(s_{i}^{\ast},\boldsymbol{s}_{-i}^{\ast})\in\boldsymbol{S}$ is a \textit{pure Nash equilibrium} if $u_{i}(s_{i}^{\ast},\boldsymbol{s}_{-i}^{\ast})\geq u_{i}(s_{i},\boldsymbol{s}_{-i}^{\ast}) \ \forall\boldsymbol{s}\in\boldsymbol{S},i\in L$. 

An \textit{exact potential game} is a game for which there exists a function $\Phi:\boldsymbol{S}\to\mathbb{R}$, called the \textit{potential function} of the game, for which it holds that 
\begin{align}\label{equation:potential_function_proof}
     \Phi(s_{k}^{\prime},\boldsymbol{s}_{-k})-\Phi(s_{k},\boldsymbol{s}_{-k})=u_{k}(s_{k}^{\prime},\boldsymbol{s}_{-k})-u_{k}(s_{k},\boldsymbol{s}_{-k})
\end{align}
for each player $k$, any two strategies $s_{k},s_{k}^{\prime}\in S_{k}$ and each strategy profile $\boldsymbol{s}_{-k}\in\boldsymbol{S}_{-k}$. Potential games are guaranteed to have pure Nash equilibria under the condition that the potential function has a maximum: Each pure equilibrium of a potential game corresponds to a local maximum of its potential function \cite{monderer1996potential}.

In this work, we show that lending games are potential games. To that end, we define the following notation. For a subset $[z]=\{1,\dots,z\}$ of the first $z$ lenders in $L$, let
\begin{align}\label{equation:interest_rate_first_z}
    r_{j}(\boldsymbol{s},z)=(r_{\text{min}}-r_{\text{max}})\frac{\sum_{i\in[z]}s_{ij}}{d_{j}}+r_{\text{max}}.
\end{align}
Using this notation, for a given strategy profile $\boldsymbol{s}$, we then define the function $\Phi(\boldsymbol{s}):\boldsymbol{S}\to\mathbb{R}$ for $G$ as
\begin{align}\label{equation:potential_function}
    \Phi(\boldsymbol{s})=\sum_{j\in B}\sum_{i\in L}(r_{j}(\boldsymbol{s},i)-r_{\text{min}})s_{ij}.
\end{align}
In the following theorem, we show that this function is a potential function for $G$, which yields that $G$ is an exact potential game.

\begin{theorem}\label{theorem:potential_game}
    The game $G$ is an exact potential game with the potential function $\Phi$.
\end{theorem}
A proof of this theorem is provided in Appendix \ref{apx:theorem:potential_game}.

\section{Equilibrium Analysis}\label{section:equilibrium}
In this section, we investigate whether a lending game $G$ has a pure Nash equilibrium, and what this equilibrium looks like. Since by Theorem \ref{theorem:potential_game}, game $G$ is an exact potential game, with potential function $\Phi$, this task boils down to studying the local maxima of $\Phi$.

\subsection{Uniqueness}
First of all, in this subsection, we prove that any lending game $G$ has a unique pure Nash equilibrium.

\begin{lemma}\label{lemma:existence}
    The potential function $\Phi$ has at least one maximiser over the strategy space $\boldsymbol{S}$.
\end{lemma}
\begin{proof} 
    It follows directly from (\ref{equation:interest_rate_first_z}) and (\ref{equation:potential_function}) that $\Phi$ is continuous. In addition, according to the definitions of $S_{i}$ (see (\ref{equation:strategy_set})), the strategy space $\boldsymbol{S}$ (i.e., the domain of $\Phi$) is clearly closed and bounded, and therefore compact. Thus, the extreme value theorem indicates that $\Phi$ has a maximum and at least one maximiser.
\end{proof}

\begin{lemma}\label{lemma:strict_concavity}
    The potential function $\Phi$ is strictly concave over $\boldsymbol{S}$.
\end{lemma}
A proof of this lemma is provided in Appendix \ref{apx:lemma:strict_concavity}.

\begin{theorem}\label{theorem:uniqueness}
    The game $G$ has a unique pure Nash equilibrium.
\end{theorem}
\begin{proof}
    By Theorem \ref{theorem:potential_game}, the game $G$ is an exact potential game with the potential function $\Phi$. In addition, $\Phi$ has at least one maximiser by Lemma \ref{lemma:existence}, while Lemma \ref{lemma:strict_concavity} proves that $\Phi$ is strictly concave, which implies that $\Phi$ has at most one maximiser. Therefore, $\Phi$ has a unique maximiser over $\boldsymbol{S}$, corresponding to its global maximum. Furthermore, since each pure Nash equilibrium corresponds to a maximum of the potential function in potential games \cite{monderer1996potential}, the unique maximiser of $\Phi$ is exactly the unique pure Nash equilibrium of $G$.
\end{proof}

\subsection{Characterising the Equilibrium}
Since each pure Nash equilibrium corresponds to a local maximum of the potential function in potential games, given that the game $G$ has a unique pure Nash equilibrium, this unique Nash equilibrium corresponds to the global maximum of the potential function $\Phi$. We can find this global maximum by solving the following constrained optimisation problem:
\begin{align}\label{equation:optimisation_problem}
    \begin{array}{rll}
        \max & \Phi(\boldsymbol{s}) & \\
        \text{s.t.} & \sum_{j\in B}s_{ij}\leq c_{i} & \quad \forall i\in L, \\
        & s_{ij}\geq 0 & \quad \forall i\in L,j\in B. \\
    \end{array}
\end{align}
Note that the constraints are all linear and simply express that $\boldsymbol{s}\in\boldsymbol{S}$.
We can characterise the optimal solution through analysing the Karush-Kuhn-Tucker (KKT) conditions. To that end, we analyse the following Lagrangian function:
\begin{align*}
    \mathcal{L}(\boldsymbol{s},\boldsymbol{\mu})=\Phi(\boldsymbol{s})+\sum_{i\in L}\mu_{i}\left(c_{i}-\sum_{j\in B}s_{ij}\right)+\sum_{i\in L}\sum_{j\in B}\mu_{ij}s_{ij}
\end{align*}
where $\boldsymbol{\mu}=(\mu_{1},\dots,\mu_{m},\mu_{11},\dots,\mu_{mn})$ is a vector of multipliers where all elements are non-negative and there are $m+mn$ elements in total. Then, the KKT conditions can be stated as follows.
\begin{enumerate}
    \item Primal feasibility:
    \begin{align*}
        \begin{cases}
            \sum_{j\in B}s_{ij}\leq c_{i} & \quad \forall i\in L, \\
            s_{ij}\geq0 & \quad \forall i\in L,j\in B. \\
        \end{cases}
    \end{align*}
    \item Stationarity: For any $i\in L,j\in B$,
    \begin{align*}
        \frac{\partial\mathcal{L}}{\partial s_{ij}}(\boldsymbol{s})=(r_{\text{min}}-r_{\text{max}})\left(\frac{1}{d_{j}}\left(s_{ij}+\sum_{k\in L}s_{kj}\right)-1\right)-\mu_{i}+\mu_{ij}=0
    \end{align*}
    where we note that the partial derivatives are straightforward to derive from the following reformulation of $\Phi$, which can be obtained from (\ref{equation:interest_rate_first_z}) and (\ref{equation:potential_function}) by rearranging the terms in the expression:
    \begin{align}\label{equation:quadratic_function}
        \resizebox{\linewidth}{!}{$
            \Phi(\boldsymbol{s})=\displaystyle\sum_{j\in B}\left(\dfrac{r_{\text{min}}-r_{\text{max}}}{2d_{j}}\left(\sum_{i\in L}s_{ij}^2+\left(\sum_{i\in L}s_{ij}\right)^2\right)+(r_{\text{max}}-r_{\text{min}})\displaystyle\sum_{i\in L}s_{ij}\right).
        $}
    \end{align}
    \item Dual feasibility:
    \begin{align*}
        \begin{cases}
            \mu_{i}\geq0 & \quad \forall i\in L, \\
            \mu_{ij}\geq0 & \quad \forall i\in L,j\in B. \\
        \end{cases}
    \end{align*}
    \item Complementary slackness:
    \begin{align*}
        \begin{cases}
            \mu_{i}(c_{i}-\sum_{j\in B}s_{ij})=0 & \quad \forall i\in L, \\
            \mu_{ij}s_{ij}=0 & \quad \forall i\in L,j\in B. \\
        \end{cases}
    \end{align*}
\end{enumerate}
Note that the objective function of the optimisation problem (i.e., the potential function $\Phi$) is differentiable and strictly concave. And the inequality constraints in (\ref{equation:optimisation_problem}) are linear, and therefore differentiable and convex. As a result, a point that satisfies the above KKT conditions is an optimal solution to (\ref{equation:optimisation_problem}).

Now we give the solution to the optimisation problem. Assume without loss of generality that the lenders are ordered increasingly according to their budget, i.e., $c_{i}\leq c_{j} \ \forall i,j\in L,i<j$. Let $c_{m+1}=\infty$ and define $\bar{m}$ as the least index in $[m]\cup\{0\}$ for which it holds that
\begin{align}\label{equation:bar_m}
    c_{\bar{m}+1}>\frac{1}{m-\bar{m}+1}\left(\sum_{j\in B}d_{j}-\sum_{\ell\in[\bar{m}]}c_{\ell}\right).
\end{align}
Note that $\bar{m}$ is well-defined because we define $c_{m+1}=\infty$. That is, $\bar{m}$ is the lowest number such that the $m-\bar{m}$ lenders with the highest budgets have budget that exceeds a $1/(m-\bar{m}+1)$ fraction of the remaining demands of the borrowers, after the $\bar{m}$ lenders with the lowest budgets each allocate their entire budget to the borrowers.

Let $\bar{L}=[\bar{m}]$ (so that $|\bar{L}|=\bar{m}$). We now define the multipliers $\boldsymbol{\mu}$ and strategy profile $s^{\ast}$ satisfying all the KKT conditions, where we make a case distinction between the low-budget lenders (i.e. those in $\bar{L}$) and the high-budget lenders (i.e., those in $L\setminus\bar{L}$). It will be clear from the definition of $\boldsymbol{s}^{\ast}$ below that the lenders in $\bar{L}$ will spend their entire budget, whereas the lenders in $L\setminus\bar{L}$ will not. Consider the following KKT multipliers:
\begin{align}\label{equation:multipliers}
    \resizebox{\linewidth}{!}{$
        \mu_{i}=
        \begin{cases}
            (r_{\text{min}}-r_{\text{max}})\left(\dfrac{c_{i}}{\sum_{k\in B}d_{k}}-\dfrac{1}{m-\bar{m}+1}\left(1-\dfrac{\sum_{\ell\in\bar{L}}c_{\ell}}{\sum_{k\in B}d_{k}}\right)\right) & \text{if } i\in\bar{L}, \\
            0 & \text{if } i\in L\setminus\bar{L}. \\
        \end{cases}
    $}
\end{align}
\begin{align*}
    \mu_{ij}=0 \quad \forall i\in L,j\in B.
\end{align*}
The strategy profile at equilibrium $\boldsymbol{s}^{\ast}$ is defined as follows. For any $i\in L,j\in B$
\begin{align}\label{equation:solution}
    s_{ij}^{\ast}
    \begin{cases}
        \dfrac{c_{i}}{\sum_{k\in B}d_{k}}\cdot d_{j} & \text{if } i\in\bar{L}, \\
        \dfrac{1}{m-\bar{m}+1}\left(1-\dfrac{\sum_{\ell\in\bar{L}}c_{\ell}}{\sum_{k\in B}d_{k}}\right)\cdot d_{j} & \text{if } i\in L\setminus\bar{L}. \\
    \end{cases}
\end{align}
The above solution is the unique pure Nash equilibrium solution to the game $G$. Theorem \ref{theorem:solution} verifies its correctness.

\begin{theorem}\label{theorem:solution}
    The strategy profile $\boldsymbol{s}^{\ast}$ defined by (\ref{equation:solution}), together with the Lagrange multipliers $\boldsymbol{\mu}$ defined by (\ref{equation:multipliers}), satisfy all the KKT conditions. Therefore, $\boldsymbol{s}^{\ast}$ is the unique pure Nash equilibrium of the interbank lending game $G$.
\end{theorem}
A proof of this theorem is provided in Appendix \ref{apx:theorem:solution}.

\begin{theorem}
    The interest rates offered by borrowing banks are identical at pure Nash equilibrium $\boldsymbol{s}^{\ast}$.
\end{theorem}
\begin{proof}
    According to the definition in (\ref{equation:interest_rate}), for any $j\in B$,
    \begin{align*}
        r_{j}(\boldsymbol{s}^{\ast}) & =(r_{\text{min}}-r_{\text{max}})\dfrac{\sum_{i\in L}s_{ij}^{\ast}}{d_{j}}+r_{\text{max}} \\
        & =(r_{\text{min}}-r_{\text{max}})\dfrac{\sum_{\ell\in\bar{L}}s_{\ell j}^{\ast}+\sum_{h\in L\setminus\bar{L}}s_{hj}^{\ast}}{d_{j}}+r_{\text{max}}.   
    \end{align*}
    Through (\ref{equation:solution}), we then obtain that for any $j\in B$,
    \begin{align*}
        \resizebox{\linewidth}{!}{$
            r_{j}(\boldsymbol{s}^{\ast})=\dfrac{r_{\text{min}}}{m-\bar{m}+1}\left(m-\bar{m}+\dfrac{\sum_{\ell\in\bar{L}}c_{\ell}}{\sum_{k\in B}d_{k}}\right)+\dfrac{r_{\text{max}}}{m-\bar{m}+1}\left(1-\dfrac{\sum_{\ell\in\bar{L}}c_{\ell}}{\sum_{k\in B}d_{k}}\right).
        $}
    \end{align*}
    It is clear that this interest rate is independent of $j$'s identity, so it is identical for any $j\in B$, which can be interpreted as a natural \textit{market interest rate} in the interbank money market.
\end{proof}

\subsection{Algorithm}
For computing the unique pure Nash equilibrium of a given lending game, we first note that the optimisation problem defined in (\ref{equation:optimisation_problem}) is an instance of a convex quadratic program, because all constraints are linear, while the objective function is quadratic (see (\ref{equation:quadratic_function})) and strictly concave (by Lemma \ref{lemma:strict_concavity}). Therefore, the optimisation problem can be reformulated as minimising a convex quadratic objective function over a polyhedron, which implies that we can use known weakly polynomial time algorithms (see \cite{kozlov1980polynomial,ye1989extension}) to find the optimal solution.

However, the characterisation in the previous section makes it easier to construct the pure Nash equilibrium of the lending game algorithmically. The natural algorithm that follows from the above characterisation first identifies the least $\bar{m}$ for which (\ref{equation:bar_m}) holds, and then computes $\bar{L}$, i.e., $\bar{L}$ are the $\bar{m}$ lenders with the lowest budgets in $L$. Finally, the algorithm sets and outputs $\boldsymbol{s}^{\ast}$ according to (\ref{equation:solution}). In Algorithm \ref{algorithm:solution}, we present a more detailed specification of the above procedure, noting that Line \ref{line:s_pi} outputs $\boldsymbol{s}^{\ast}$ permuted back to the original indexing of $L$.

\begin{algorithm}
    \caption{Algorithm to find pure Nash equilibrium of $G$.}
    \label{algorithm:solution}
    \begin{algorithmic}[1]
        \REQUIRE Lending game $G = (m,n,\boldsymbol{c},\boldsymbol{d},r_{\text{max}},r_{\text{min}})$
        \ENSURE Unique pure Nash equilibrium $\boldsymbol{s}^{\ast}$
        \STATE Without loss of generality, rename the lenders $L$ such that $c_{1}\leq c_{2}\leq\cdots\leq c_m$ holds. Let $\pi:[m]\to[m]$ be the corresponding permutation of $L$.
        \STATE Let $c_{m+1}:=\infty$, $R:=\sum_{j\in B}d_{j}/(m+1)$, and $i:=0$. \label{line:R}
        \WHILE{$c_{i+1}\leq R$}
            \STATE Set $R:=(R\cdot(m-i+1)-c_{i+1})/(m-i)$.
            \STATE Set $i:=i+1$.
        \ENDWHILE
        \STATE Set $\bar{m}:=i$. \label{line:bar_m}
        \STATE Set $\bar{L}=[\bar{m}]$, and using this choice of $\bar{L}$, set $\boldsymbol{s}^{\ast}$ as in (\ref{equation:solution}). \label{line:s_star}
        \STATE Output $\boldsymbol{s}^{\pi}$ where $s_{ij}^{\pi}=s_{\pi^{-1}({i})j}^{\ast} \ \forall i\in L,j\in B$. \label{line:s_pi}
    \end{algorithmic}
\end{algorithm}

\begin{theorem}\label{theorem:correctness}
    Algorithm \ref{algorithm:solution} correctly outputs the unique pure Nash equilibrium of a given lending game, and has a (strongly polynomial) running time of $\mathcal{O}(mn+m\log m)$.
\end{theorem}
\begin{proof}
    If we assume that $\bar{m}$ computed by Algorithm \ref{algorithm:solution} in Line \ref{line:bar_m} coincides with the definition of $\bar{m}$ given by (\ref{equation:bar_m}), then it follows from Lines \ref{line:s_star} and \ref{line:s_pi} that the strategy profile defined in (\ref{equation:solution}) is output by Algorithm \ref{algorithm:solution}. Correctness of Algorithm \ref{algorithm:solution} then follows immediately from Theorem \ref{theorem:solution}. Therefore, what remains is to prove that indeed $\bar{m}$ is computed correctly: Denote by $\text{RHS}_{i}$ the value on the right-hand side of (\ref{equation:bar_m}) if we would substitute $i$ in the place of $\bar{m}$. We now see that $R$ defined in Line \ref{line:R} coincides with $\text{RHS}_{0}$. For any $i>0$, we observe that $\text{RHS}_{i+1}=(\text{RHS}_{i}\cdot(m-i+1)-c_{i+1})/(m-i)$, which implies that after the $i$-th iteration of the while-loop, it holds that $R=\text{RHS}_{i}$. The while-loop terminates at the smallest $i$ for which it holds that $c_{i+1}>R=\text{RHS}_{i}$, and subsequently $\bar{m}$ is set to $i$, so that after Line \ref{line:bar_m}, it holds that $c_{\bar{m}+1}>\text{RHS}_{\bar{m}}$. We conclude that $\bar{m}$ is computed correctly, and that hence Algorithm \ref{algorithm:solution} outputs the unique pure Nash equilibrium correctly.

    For the running time, the renaming step (or the computation of $\pi$) requires sorting $m$ values, which can be done in $\mathcal{O}(m\log m)$ time. Finding the correct choice of $\bar{m}$ (Lines \ref{line:R} to \ref{line:bar_m}) requires computing $\sum_{j\in B}d_{j}$ which takes $\mathcal{O}(n)$ time, and going through at most $m$ iterations of the while-loop, where each iteration takes constant time.
    Lastly, in Line \ref{line:s_star}, the values of $s_{ij}^{\ast}$ are computed for all $i\in L,j\in B$, according to (\ref{equation:solution}). This can be done efficiently by pre-computing $\sum_{\ell\in \bar{L}}c_{\ell}$ (in $\mathcal{O}(m)$ time) and using the pre-computed value $\sum_{j\in B}d_{j}$. Computing the correct value of $s_{ij}^{\ast}$ given by the expression in (\ref{equation:solution}) then takes constant time for each $i\in L,j\in B$. 
    Thus, Line \ref{line:s_star} takes $\mathcal{O}(mn)$ time, and therefore, the overall time complexity of Algorithm \ref{algorithm:solution} is $\mathcal{O}(mn+m\log m)$.
\end{proof}

\section{Dynamics Analysis}\label{section:dynamics}
A natural approach to modelling players' decision-making process is the best-response dynamics, where players iteratively update their strategies by choosing the best-response to the current strategy profile. We investigate lender behaviour over time in the lending game $G$. Given that $G$ is an exact potential game with a strictly concave potential function $\Phi$, we expect that some variants of the best-response dynamics converge to the pure Nash equilibrium computed in the previous section. We proceed by formally establishing this result in both discrete-time and continuous-time scenarios.

\textbf{Best responses.} Formally, given a strategy profile $\boldsymbol{s}^t\in\boldsymbol{S}$ at time $t$, for each lender $i\in L$, its \textit{best-response strategy set} is defined as
\begin{align*}
    \text{BR}_{i}(\boldsymbol{s}^t)=\arg\max\{u_{i}(z_{i},\boldsymbol{s}_{-i}^t):z_{i}\in S_{i}\}
\end{align*}
where $\boldsymbol{s}_{-i}^t\in\boldsymbol{S}_{-i}$ denotes the current strategies of other lenders at time $t$. In addition, since $G$ is an exact potential game with potential function $\Phi$, we can reformulate this best-response strategy set as
\begin{align*}
    \text{BR}_{i}(\boldsymbol{s}^t)=\arg\max\{\Phi(z_{i},\boldsymbol{s}_{-i}^t):z_{i}\in S_{i}\}.
\end{align*}

\subsection{Discrete-Time Dynamics}
Given a current strategy profile $\boldsymbol{s}^t$, to define discrete-time best-response dynamics, we consider the following strategy update for a lender $i\in L$,
\begin{align*}
    s_i^{t+1}= s_i^t+\alpha_t(\hat{s}_i^t-s_i^t), \quad \hat{s}_i^t\in\text{BR}_{i}(\boldsymbol{s}^t)
\end{align*}
where $\alpha_t\in(0,1]$ represents the step-size of the update at time $t$. Note that this generalises the standard definition of best-response dynamics that is encountered in most literature, where $\alpha_t=1$ for all $t\in\mathbb{N}$. This generalisation makes sense to consider in our context, due to the lending game having a continuous strategy space and the concavity of the utility functions. In this subsection, we investigate the discrete-time best-response dynamics of the lending game $G$ in asynchronous and synchronous update patterns, respectively.

\subsubsection{Asynchronous Updates}
Asynchronous updates allow a lender at a time to respond. Concretely, the trajectory of the \textit{asynchronous best-response dynamics} for the lending game $G$ is a sequence $(\boldsymbol{s}^t)_{t\in\mathbb{N}}$ of strategy profiles, where $\boldsymbol{s}^1$ is the \textit{initial strategy profile}, and for all $t\in\mathbb{N}$, it holds that there is a lender $i^t\in L$ such that $\boldsymbol{s}^{t+1}=(s_{i^t}^{t+1},\boldsymbol{s}_{-i^t}^t)$ and $s_{i^t}^{t+1}=s_{i^t}^t+\alpha_t(\hat{s}_{i^t}^{t}-s_{i^t}^t)$ where $\hat{s}_{i^t}^t\in\text{BR}_{i^t}(\boldsymbol{s}^t)$. We refer to the case where $\alpha_t$ is set to a uniform value $\alpha\in(0,1]$ for all $t\in\mathbb{N}$ as \textit{$\alpha$-uniform best-response dynamics}. We then consider the following two variants of $\alpha$-uniform best-response dynamics.
\begin{itemize}
    \item \textit{Eager best-response dynamics}: In this variant of best-response dynamics, at every time step $t$, the updating lender $i^t$ is defined as the one who updates its strategy to achieve the highest utility increase among all lenders $L$, breaking ties in an arbitrary deterministic way, i.e., $i^t\in\arg\max_{i\in L}\{u_{i}(\hat{s}_i^t,\boldsymbol{s}_{-i}^t):\hat{s}_i^t\in\text{BR}_{i}(\boldsymbol{s}^t)\}$. This can be interpreted as the fact that the lender who is most ``eager'' to improve its utility updates its strategy first. Alternatively, this dynamics can be interpreted from the viewpoint of a centralised coordination which selects the lender $i^t$ whose best-response can yield the highest utility increase at each time step $t$.
    \item \textit{Randomised best-response dynamics}: In this variant of best-response dynamics, the updating lender $i^t$ is chosen at random at every time step $t$, according to a given probability distribution $\pi^t$ on the set of lenders $L$. This probability distribution must satisfy that there exists a constant $p_\text{min}>0$ such that each lender at each time step has a probability of at least $p_\text{min}$ of being selected. For example, taking $\pi^t$ to be the uniform distribution on $L$ satisfies this property, and we would then have $p_\text{min}= 1/m$.
\end{itemize}
\begin{theorem}\label{theorem:eager_alpha_uniform_best_response_dynamics}
    For any lending game $G$, any $\alpha\in(0,1]$, and any initial strategy profile, the eager $\alpha$-uniform best-response dynamics converges to the unique pure Nash equilibrium $\boldsymbol{s}^{\ast}$.
\end{theorem}
\begin{proof}
    Let $G=(m,n,\boldsymbol{c},\boldsymbol{d},r_{\text{max}},r_{\text{min}})$ be a lending game with potential function $\Phi$ defined by (\ref{equation:potential_function}), let $\boldsymbol{s}^1$ and $\boldsymbol{s}^{\ast}$ be any initial strategy profile and the pure Nash equilibrium (i.e., the strategy profile maximising $\Phi$), respectively, and let $\alpha$ be an arbitrary value in $(0,1]$. Consider the eager $\alpha$-uniform best-response dynamics $(\boldsymbol{s}^t)_{t\in\mathbb{N}}$ defined above. Since $G$ is an exact potential game with potential function $\Phi$, we obtain that for any $t\in\mathbb{N}$,
    \begin{align*}
        u_{i^t}(s_{i^t}^{t+1},\boldsymbol{s}_{-i^t}^t)-u_{i^t}(s_{i^t}^t,\boldsymbol{s}_{-i^t}^t)=\Phi(s_{i^t}^{t+1},\boldsymbol{s}_{-i^t}^t)-\Phi(s_{i^t}^t,\boldsymbol{s}_{-i^t}^t).
    \end{align*}
    If $\boldsymbol{s}_{i^t}^t\not\in\text{BR}_{i^t}(\boldsymbol{s}^t)$, $u_{i^t}(\boldsymbol{s}_{i^t}^{t+1},\boldsymbol{s}_{-i^t}^t)>u_{i^t}(\boldsymbol{s}_{i^t}^t,\boldsymbol{s}_{-i^t}^t)$, so $\Phi(\boldsymbol{s}^{t+1})>\Phi(\boldsymbol{s}^t)$. If $\boldsymbol{s}_{i^t}^t\in\text{BR}_{i^t}(\boldsymbol{s}^t)$, $\Phi(\boldsymbol{s}^{t+1})=\Phi(\boldsymbol{s}^t)$. Thus, the sequence $(\Phi(\boldsymbol{s}^t))_{t\in\mathbb{N}}$ is non-decreasing. Considering that $\Phi$ is bounded from above by $\Phi(\boldsymbol{s}^{\ast})$, $(\Phi(\boldsymbol{s}^t))_{t\in\mathbb{N}}$ converges. Let $I$ be the limit, i.e.,
    \begin{align}\label{equation:limit}
        \lim_{t\to\infty}\Phi(\boldsymbol{s}^t)=I\leq\Phi(\boldsymbol{s}^{\ast}).
    \end{align}
    We show that the latter inequality in fact holds with equality, which will yield our claim.

    For proving that $I=\Phi(\boldsymbol{s}^{\ast})$, we first derive some properties relating to the gradient and Hessian of $\Phi$. First, we can straightforwardly derive the first-order and second-order partial derivatives of $\Phi$ by using the form (\ref{equation:quadratic_function}) of $\Phi$: In particular, for the first-order derivatives, 
    \begin{align*} 
        (\nabla\Phi(\boldsymbol{s}))_{ij}=(r_{\text{min}}-r_{\text{max}})\left(\frac{1}{d_{j}}\left(s_{ij}+\sum_{k\in L}s_{kj}\right)-1\right) 
    \end{align*}
    for all $(i,j)\in L\times B$, 
    and for the second-order derivatives, we have
    \begin{align*}
        (\boldsymbol{H}\Phi(\boldsymbol{s}))_{ij,k\ell}=
        \begin{cases}
            0 & \text{if } k\neq i \text{ and } \ell\neq j \\
            (r_{\text{min}}-r_{\text{max}})/d_{j} & \text{if } k\neq i \text{ and } \ell=j \\
            2(r_{\text{min}}-r_{\text{max}})/d_{j} & \text{if } k=i \text{ and } \ell=j 
        \end{cases}
    \end{align*}
    for all $((i,j),(k,\ell))\in(L\times B)^2$. For convenience, let
    \begin{align*}
        a=2(r_{\text{max}}-r_{\text{min}})/d_{j}.
    \end{align*}
    From the Hessian, it follows that $\nabla\Phi$ decreases everywhere at a rate of at most $a$. Therefore, if $\boldsymbol{s}$ is a point such that $(\nabla\Phi(\boldsymbol{s}))_{ij}>0$ for some coordinate $(i,j)$, then every point $\boldsymbol{s}^\prime$ that lies in the $\ell_1$-ball of radius $\Phi(\boldsymbol{s})/a$ around $\boldsymbol{s}$ also satisfies that $(\nabla\Phi(\boldsymbol{s}^\prime))_{ij}>0$. We can strengthen and generalise the latter insight to better lower bounds, and arbitrary linear combinations over the entries of $\nabla \Phi$, as follows.
    \begin{lemma}\label{lemma:gradient_bound}
        Let $\boldsymbol{v}\in\mathbb{R}^{mn}$ and suppose that $\boldsymbol{v}\cdot\nabla\Phi(\boldsymbol{s})>0$, then every point $\boldsymbol{s}^\prime$ in $\ell_1$-ball of radius $\boldsymbol{v}\cdot\nabla \Phi(\boldsymbol{s})/(2a\max_{i,j}|v_{ij}|)$ around $\boldsymbol{s}$ satisfies $\boldsymbol{v}\cdot\nabla\Phi(\boldsymbol{s}^\prime)\geq(1/2)\boldsymbol{v}\cdot\nabla\Phi(\boldsymbol{s})$.
    \end{lemma}
    A proof of this lemma is provided in Appendix \ref{apx:lemma:gradient_bound}.
    
    The following lemma forms the key to proving our claim by using Lemma \ref{lemma:gradient_bound}.
    \begin{lemma}\label{lemma:large_deviation}
        Let $\boldsymbol{s}$ be any strategy profile of $G$. There exists a lender $i\in L$ who can deviate from $s_i$ to $s_i^\prime$ such that $u_{i}(s_i^\prime,\boldsymbol{s}_{-i})-u_{i}(\boldsymbol{s})=\Phi(s_i^\prime,\boldsymbol{s}_{-i})-\Phi(\boldsymbol{s})\geq(\Phi(\boldsymbol{s}^{\ast})-\Phi(\boldsymbol{s}))^2/(4m^4n^2a(\max_{k\in L}c_k)^2)$.
    \end{lemma}
    \begin{proof} 
        The smoothness of the potential function $\Phi$ guarantees that the directional derivative of $\Phi$ exists everywhere, and along any vector. In particular, consider the directional derivative $\nabla_{\boldsymbol{v}}\Phi(\boldsymbol{s})$ at $\boldsymbol{s}$ in the direction of $\boldsymbol{s}^{\ast}$ (i.e., $\boldsymbol{v}$ is a unit vector in the direction $\boldsymbol{s}^{\ast}-\boldsymbol{s}^t$). By concavity of $\Phi$, $\nabla_{\boldsymbol{v}}\Phi(\boldsymbol{s})>0$, that is,
        \begin{align*}
            \nabla_{\boldsymbol{v}}\Phi(\boldsymbol{s})=\boldsymbol{v}\cdot\nabla\Phi(\boldsymbol{s})=\sum_{i,j}v_{ij}(\nabla\Phi(\boldsymbol{s}))_{ij}>0.
        \end{align*}
        By partitioning the above summation into $m$ parts, by lender, we conclude that for at least one lender $i\in L$, it holds that
        \begin{align*}
            \sum_{j\in B}v_{ij}(\nabla\Phi(\boldsymbol{s}))_{ij}>\frac{\boldsymbol{v}\cdot\nabla\Phi(\boldsymbol{s})}{m}.
        \end{align*}
        Then, define $\boldsymbol{v}^\prime$ by
        \begin{align*}
            v_{kj}^\prime=\left\{
            \begin{array}{ll}
                v_{ij} & \text{if } k=i, \\
                0 & \text{if } k\neq i, \\
            \end{array}
            \right. \quad \forall j\in B.
        \end{align*}
        Let $\boldsymbol{w}$ be the unit vector in the direction of $\boldsymbol{v}^\prime$ (i.e., $\boldsymbol{w}$ is an appropriate scaling of $\boldsymbol{v}^\prime$). As the length of $\boldsymbol{w}$ exceeds that of $\boldsymbol{v}^\prime$, it follows that $\boldsymbol{w}\cdot\nabla\Phi(\boldsymbol{s})>\boldsymbol{v}\cdot\nabla\Phi(\boldsymbol{s})/m$.
    
        Now, consider lender $i$'s deviation from $s_i$ to $s_i^\prime$, where for $j\in B$, $s_{ij}^\prime=s_{ij}+(\boldsymbol{v}\cdot\nabla\Phi(\boldsymbol{s})/(2ma\max_{i,j}|w_{ij}|))w_{ij}$, resulting in the point
        \begin{align*}
            \boldsymbol{s}^\prime=(s_i^\prime,\boldsymbol{s}_{-i})=\boldsymbol{s}+\frac{\boldsymbol{v}\cdot\nabla\Phi(\boldsymbol{s})}{2ma\max_{i,j}|w_{ij}|}\boldsymbol{w}.
        \end{align*}
        We observe that since $\|\boldsymbol{w}\|=\sum_{j\in B}|w_{ij}|=1$,
        \begin{align*}
            \|\boldsymbol{s}^\prime-\boldsymbol{s}\|=\frac{\boldsymbol{v}\cdot\nabla\Phi(\boldsymbol{s})}{2ma\max_{i,j}|w_{ij}|}\|\boldsymbol{w}\|<\frac{\boldsymbol{w}\cdot\nabla\Phi(\boldsymbol{s})}{2a\max_{i,j}|w_{ij}|},
        \end{align*}
        so $\boldsymbol{s}^\prime$ lies in the $\ell_1$-ball of radius $\boldsymbol{w}\cdot\nabla\Phi(\boldsymbol{s})/(2a \max_{i,j}|w_{ij}|)$ around $\boldsymbol{s}$, and hence by Lemma \ref{lemma:gradient_bound}, for every point $\boldsymbol{s}^{\prime\prime}$ on the line segment between $\boldsymbol{s}$ and $\boldsymbol{s}^\prime$, it holds that $\boldsymbol{w}\cdot\nabla\Phi(\boldsymbol{s}^{\prime\prime})\geq(1/2)\boldsymbol{w}\cdot\nabla\Phi(\boldsymbol{s})$. Thus, by integrating along the line segment from $\boldsymbol{s}$, in the direction $\boldsymbol{w}$, up to the point $\boldsymbol{s}^\prime$, and using the latter lower bound of $\nabla_{\boldsymbol{w}}\Phi(\boldsymbol{s})$ on this line segment, we obtain that the increase $\Phi(\boldsymbol{s}^\prime)-\Phi(\boldsymbol{s})$ can be bounded as follows. 
        \begin{align}
            & \Phi(\boldsymbol{s}^\prime)-\Phi(\boldsymbol{s})\geq\frac{\boldsymbol{w}\cdot\nabla\Phi(\boldsymbol{s})}{2}\|\boldsymbol{s}^\prime-\boldsymbol{s}\| \notag \\ 
            & \qquad =\frac{\boldsymbol{w}\cdot\nabla\Phi(\boldsymbol{s})}{2}\cdot\frac{\boldsymbol{v}\cdot\nabla\Phi(\boldsymbol{s})}{2ma\max_{i,j}|w_{ij}|}\|\boldsymbol{w}\| \notag \\
            & \qquad =\frac{(\boldsymbol{w}\cdot\nabla\Phi(\boldsymbol{s}))(\boldsymbol{v}\cdot\nabla\Phi(\boldsymbol{s}))}{4ma\max_{i,j}|w_{ij}|}>\frac{(\boldsymbol{v}\cdot\nabla\Phi(\boldsymbol{s}))^2}{4m^2a\max_{i,j}|w_{ij}|} \label{equation:increase_bound}.
        \end{align}
        
        Next, we derive a lower bound of $\boldsymbol{v}\cdot\nabla\Phi(\boldsymbol{s})$. Consider the line segment $\mathbb{L}$ which has $\boldsymbol{s}$ and $\boldsymbol{s}^{\ast}$ as its start-points and end-points. The potential function $\Phi$ increases from $\Phi(\boldsymbol{s})$ to $\Phi(s^{\ast})$ on $\mathbb{L}$, and the length of $\mathbb{L}$ is at most $mn\max_{k\in L}c_k$. By concavity of $\Phi$, the directional derivative at $\boldsymbol{s}$ in the direction of $\boldsymbol{v}$ must be at least the average increase-per-unit of $\Phi$ along $\mathbb{L}$, which is at least $(\Phi(\boldsymbol{s}^{\ast})-\Phi(\boldsymbol{s}))/(mn\max_{k\in L}c_k)$. In other words,
        \begin{align*}
            \boldsymbol{v}\cdot\nabla\Phi(\boldsymbol{s})\geq\frac{\Phi(\boldsymbol{s}^{\ast})-\Phi(\boldsymbol{s})}{mn\max_{k\in L}c_k}.
        \end{align*}
        Plugging this into (\ref{equation:increase_bound}), we obtain that the lender $i$ can beneficially deviate at $\boldsymbol{s}$ by changing their strategy to $\boldsymbol{s}_i'$ and thereby improving their utility (and thus the potential function $\Phi$) by an amount of 
        \begin{align*}
            \frac{(\Phi(\boldsymbol{s}^{\ast})-\Phi(\boldsymbol{s}))^2}{4m^4n^2a(\max_{i,j}|w_{ij}|)(\max_{k\in L}c_k)^2}\geq\frac{(\Phi(\boldsymbol{s}^{\ast})-\Phi(\boldsymbol{s}))^2}{4m^4n^2a(\max_{k\in L}c_k)^2}.
        \end{align*}
    \end{proof}
    
    We are now ready to show that (\ref{equation:limit}) holds with equality. Suppose for contradiction that $I<\Phi(\boldsymbol{s}^{\ast})$. Let $\delta>0$ be a sufficiently small value such that
    \begin{align*}
        \delta<\frac{\alpha}{2}\cdot\frac{(\Phi(\boldsymbol{s}^{\ast})-I)^2}{4m^4n^2a(\max_{k\in L}c_k)^2}.
    \end{align*}
    Let $t$ be a time step such that $\Phi(\boldsymbol{s}^t)>I-\delta$. By Lemma \ref{lemma:large_deviation}, there exists a lender $i\in L$ for which every best-response $\hat{s}_i^t\in\text{BR}_{i}(\boldsymbol{s}^t)$ satisfies $\Phi(\hat{s}_i^t,\boldsymbol{s}_{-i}^t)-\Phi(\boldsymbol{s}^t)\geq2\delta/\alpha$. By the definition of eager best-response dynamics, it must also hold for the updating lender $i^t$ that $\Phi(\hat{s}_{i^t}^t,\boldsymbol{s}_{-i^t}^t)-\Phi(\boldsymbol{s}^t)\geq2\delta/\alpha$. Let $z_{i^t}\in\{s_{i^t}^t+\alpha(\hat{s}_{i^t}^t-s_{i^t}^t):\hat{s}_{i^t}^t\in\text{BR}_{i^t}(\boldsymbol{s}^t)\}$, it holds that $z_{i^t}=(1-\alpha)s_{i^t}^t+\alpha\hat{s}_{i^t}^t$ for some $\hat{s}_{i^t}^t\in\text{BR}_{i^t}(\boldsymbol{s}^t)$. By concavity of $\Phi$, we obtain
    \begin{align*}
        \Phi(z_{i^t},\boldsymbol{s}_{-i^t}^t)-\Phi(\boldsymbol{s}^t) & =\Phi((1-\alpha)\boldsymbol{s}^t+\alpha(\hat{s}_{i^t}^t,\boldsymbol{s}_{-i^t}^t)) - \Phi(\boldsymbol{s}^t) \\
        & \geq\alpha(\Phi(\hat{s}_{i^t}^t,\boldsymbol{s}_{-i^t}^t) -\Phi(\boldsymbol{s}^t)) \geq 2\delta.
    \end{align*}
    Thus, given that $\boldsymbol{s}^{t+1}_i\in\{s_i^t+\alpha(\hat{s}_i^t-s_i^t):\hat{s}_i^t\in\text{BR}_{i}(\boldsymbol{s}^t)\}$, we have that $\Phi(s^{t+1})\geq\Phi(s^t)+2\delta\geq I-\delta+2\delta>I$, which is a contradiction. This proves our claim for eager $\alpha$-uniform best-response dynamics.
\end{proof}

\begin{theorem}\label{theorem:randomised_alpha_uniform_best_response_dynamics}
    For any lending game $G$, any $\alpha\in(0,1]$, and any initial strategy profile, the randomised $\alpha$-uniform best-response dynamics converges to the pure Nash equilibrium $\boldsymbol{s}^{\ast}$ almost surely.
\end{theorem}
\begin{proof}
    For the randomised $\alpha$-uniform best-response dynamics, we can take the same approach as in the proof for the eager $\alpha$-uniform best-response dynamics (Theorem \ref{theorem:eager_alpha_uniform_best_response_dynamics}); however, we need to consider the difference that the sequence of strategies $\boldsymbol{s}^t$ is random instead of deterministic. Here, we have that the sequence $(\Phi(\boldsymbol{s}^t))_{t\in\mathbb{N}}$ is always non-decreasing and bounded from above by $\Phi(\boldsymbol{s}^{\ast})$ so $\lim_{t\to\infty}\Phi(\boldsymbol{s}^t)$ always exists, and therefore it suffices to show that $\boldsymbol{Pr}[\lim_{t\to\infty}\Phi(\boldsymbol{s}^t)=\Phi(\boldsymbol{s}^{\ast})]=1$. 
    
    For strategy profile $\boldsymbol{s}$, let $i^{\ast}(\boldsymbol{s})\in\arg\max_{i\in L}\{\Phi(\hat{s}_i):\hat{s}_i\in\text{BR}_{i}(\boldsymbol{s})\}$. In a randomised best-response dynamics, by definition there exists $p_\text{min}>0$ such that at every time step, every lender has a probability of at least $p_\text{min}$ of being selected as the updating lender, so the (random) set $\{t\in\mathbb{N}:i^t=i^{\ast}(\boldsymbol{s}^t)\}$ is of infinite size with probability $1$. Let $(\boldsymbol{z}^t)_{t\in\mathbb{N}}$ be any realisation of the randomised best-response dynamics that belongs to the latter event, i.e., for which $i^t=i^{\ast}(\boldsymbol{s}^{\ast})$ for infinitely many $t$ (where $i^t$ is now with respect to the concrete sequence $\boldsymbol{z}$). Letting $\lim_{t\to\infty}\Phi(\boldsymbol{z}^t)=I$, we show that $I=\Phi(\boldsymbol{s}^{\ast})$. For contradiction, suppose that $I<\Phi(\boldsymbol{s}^{\ast})$, let $0<\delta<\alpha(\Phi(\boldsymbol{s}^{\ast})-I)^2/(8m^4n^2a(\max_{k\in L}c_k)^2)$, and let $t$ be a time step such that both $\Phi(\boldsymbol{z}^t)>I-\delta$ and $i^t=i^{\ast}(\boldsymbol{z}^t)$.
    By Lemma \ref{lemma:large_deviation}, there exists a lender $i$ for which every strategy $\hat{z}_i\in\text{BR}_{i}(\boldsymbol{z}^t)$ satisfies $\Phi(\hat{z}_i,\boldsymbol{z}_{-i}^t) - \Phi(\boldsymbol{z}^t) \geq 2\delta/\alpha$. By the fact that $i^t=i^{\ast}(\boldsymbol{z}^t)$, it must thus also hold for lender $i^t$ that $\Phi(\hat{z}_{i^t},\boldsymbol{z}_{-i^t}^t)-\Phi(\boldsymbol{z}^t) \geq 2\delta/\alpha$. Thus, if we let $y_{i^t}\in\{z_{i^t}^t+\alpha(\hat{z}_{i^t}^t-z_{i^t}^t):\hat{z}_{i^t}^t\in\text{BR}_{i^t}(\boldsymbol{z}^t)\}$, it holds that $y_{i^t}=(1-\alpha)z_{i^t}^t+\alpha\hat{z}_{i^t}^t$ for some $\hat{z}_{i^t}^t\in\text{BR}_{i^t}(\boldsymbol{z}^t)$. By concavity of $\Phi$, we obtain
    \begin{align*}
        \Phi(y_{i^t},\boldsymbol{z}_{-i^t}^t)-\Phi(\boldsymbol{z}^t) & =\Phi((1-\alpha)\boldsymbol{z}^t+\alpha(\hat{z}_{i^t}^t,\boldsymbol{z}_{-i^t}^t)) - \Phi(\boldsymbol{z}^t) \\
        & \geq\alpha(\Phi(\hat{z}_{i^t}^t,\boldsymbol{z}_{-i^t}^t) -\Phi(\boldsymbol{z}^t)) \geq 2\delta.
    \end{align*}
    Thus, given that $\boldsymbol{z}^{t+1}_i\in\{z_{i}^t+\alpha(\hat{z}_i^t-z_{i}^t):\hat{z}_i^t\in\text{BR}_{i}(\boldsymbol{z}^t)\}$, we have that $\Phi(z^{t+1})\geq\Phi(z^t)+2\delta\geq I-\delta+2\delta>I$, which is a contradiction and shows that any realisation $(\boldsymbol{z}^t)_{t\in\mathbb{N}}$ of the randomised best-response dynamics for which it holds that $i^t=i^{\ast}(\boldsymbol{z}^t)$ infinitely often, satisfies that $\lim_{t\to\infty}\Phi(\boldsymbol{z}^t)=\Phi(\boldsymbol{s}^{\ast})$. Thus, $\boldsymbol{Pr}[\lim_{t \to\infty} \Phi(\boldsymbol{s}^t)=\Phi(\boldsymbol{s}^{\ast})]\geq\boldsymbol{Pr}[i^t=i^{\ast}(\boldsymbol{s}^t) \text{ inf. often.}]=1$, which proves the claim for randomised $\alpha$-uniform best-response dynamics.
\end{proof}

\subsubsection{Synchronous Updates}
Rosen \cite{rosen1965existence} provides a result on the convergence of a certain discrete synchronous update dynamics obtained by discretising a continuous \textit{pseudo-gradient dynamics}, in which every player moves continuously into the direction of the gradient (projected on $\boldsymbol{S}$) of their utility function. Convergence to the pure Nash equilibrium has been shown to hold for this dynamics in case the game satisfies a property called \textit{diagonal strict concavity}. The pseudo-gradient dynamics is related to best-response dynamics in the sense that players continuously and simultaneously change their strategy in the best ``local'' direction, as prescribed by the gradient of their utility functions. Note that the continuous pseudo-gradient dynamics differs from what is known as \textit{continuous best-response dynamics}, which we consider separately in Section \ref{section:continuous_time_best_response_dynamics}.

\begin{definition}[Pseudo-gradient]
    For a lending game $G$ in which the lenders have utility functions $(u_1(\boldsymbol{s}),\dots,u_m(\boldsymbol{s}))$, and a positive weight vector $\boldsymbol{r}=(r_1,\dots,r_m)\in\mathbb{R}^m_{>0}$,
    the \textit{$\boldsymbol{r}$-pseudo-gradient} $g:\boldsymbol{S}\to\mathbb{R}^{mn}$ is defined as
    \begin{align*}
        g(\boldsymbol{s},\boldsymbol{r}):=(r_1\nabla_1u_1(\boldsymbol{s}),\dots,r_m\nabla_mu_m(\boldsymbol{s}))
    \end{align*}
    where $\nabla_iu_i(\boldsymbol{s})$ denotes the gradient of $u_{i}(\boldsymbol{s})$ on $s_i$ for $i\in L$.
\end{definition}

A \textit{discretised pseudo-gradient dynamics} is then defined as a dynamics $(\boldsymbol{s}^t)_{t\in\mathbb{N}}$ with any initial strategy profile $\boldsymbol{s}^1\in\boldsymbol{S}$ such that $\boldsymbol{s}^{t+1}=\boldsymbol{s}^t+\alpha_t g^\prime(\boldsymbol{s}^t,\boldsymbol{r})$, where $g^\prime$ denotes a projection of the $\boldsymbol{r}$-pseudo-gradient (for some $\boldsymbol{r}\in\mathbb{R}_{\geq0}^m$) such that moving in the direction of $g^\prime(\boldsymbol{s}^t,\boldsymbol{r})$ is feasible without moving outside the strategy space $\boldsymbol{S}$. Here, $(\alpha_t)_{t \in \mathbb{N}}$ are appropriately chosen sufficiently small step-sizes. Informally, this means that at each step, every player moves a small $\alpha_t$ step in the direction of steepest ascent of their utility function, restricted to not moving outside $\boldsymbol{S}$. For a precise and formal definition of the dynamics, we refer the reader to \cite{rosen1965existence}.

\begin{theorem}[Rosen \cite{rosen1965existence}, Theorem 10]
    Let $\boldsymbol{r}\in\mathbb{R}_{\geq0}^m$ and let $\nabla g(\boldsymbol{s},\boldsymbol{r})$ denote the Jacobian matrix of the $\boldsymbol{r}$-pseudo-gradient $g(\boldsymbol{s},\boldsymbol{r})$. If $\nabla g(\boldsymbol{s},\boldsymbol{r})+(\nabla g(\boldsymbol{s},\boldsymbol{r}))^\top$ is negative definite for all $\boldsymbol{s}\in\boldsymbol{S}$, then there exists $(\alpha_t)_{t\in\mathbb{N}}$ for which the discretised pseudo-gradient dynamics converge to the pure Nash equilibrium from any initial strategy profile.
\end{theorem}

In particular, when $\boldsymbol{r}=\boldsymbol{1}$, the $\boldsymbol{1}$-pseudo-gradient $g(\boldsymbol{s},\boldsymbol{1})$ is equivalent to the concatenation of standard gradients of lenders' utility functions. For simplicity, let $g(\boldsymbol{s}):=g(\boldsymbol{s},\boldsymbol{1})$. More specifically, in the lending game $G$, the above projection $g^\prime$ can be expressed by
\begin{align*}
    g_i^\prime(\boldsymbol{s}^t)=g_i(\boldsymbol{s}^t)-\mu_i \quad \forall i\in L.
\end{align*}
We establish that for $\boldsymbol{r}=\boldsymbol{1}$ the required condition on the Jacobian holds, so by the above theorem, for appropriate step-sizes $(\alpha_t)_{t\in\mathbb{N}}$, the discretised pseudo-gradient dynamics converges to the unique Nash equilibrium from any initial strategy profile.

\begin{proposition}
    For a lending game $G$, let $\nabla g(\boldsymbol{s})$ be the Jacobian matrix of its $\boldsymbol{1}$-pseudo-gradient $g(\boldsymbol{s})$ at strategy profile $\boldsymbol{s}$. It holds that $\nabla g(\boldsymbol{s})+(\nabla g(\boldsymbol{s}))^\top$ is negative definite.
\end{proposition}
A proof of this proposition is provided in Appendix \ref{apx:prop47} . This proposition establishes that there exist step-sizes for which the discretised pseudo-gradient dynamics related to the discrete-time best-response dynamics, using the $\boldsymbol{1}$-pseudo-gradient, converges to the unique pure Nash equilibrium.

\subsection{Continuous-Time Dynamics}\label{section:continuous_time_best_response_dynamics}
Now we consider the continuous-time scenario. Given the current strategy profile $\boldsymbol{s}^t$, continuous-time best-response dynamics is defined by the following differential equations \cite{hofbauer2006best,barron2010best,leslie2020best}.
\begin{align*}
    \dfrac{ds_i^t}{dt}\in \hat{s}^t_i-s_i^t, \quad \hat{s}^t_i\in\text{BR}_{i}(\boldsymbol{s}^t) \quad \forall i\in L.
\end{align*} 
Note that due to the strict concavity of $\Phi$, $\text{BR}_{i}(\boldsymbol{s})$ is a singleton for all $\boldsymbol{s} \in \boldsymbol{S}$, which implies that the right hand side is a continuous function. Furthermore, $\text{BR}_{i}(\boldsymbol{s})$ is in fact only a function of $\boldsymbol{s}_{-i}$, which implies that the right hand side is continuously differentiable, so that the dynamics is well-defined.

\balance

\begin{theorem}
    The continuous-time best-response dynamics converges to the unique pure Nash equilibrium $\boldsymbol{s}^{\ast}$ from any initial state.
\end{theorem}
\begin{proof}
    We apply Lyapunov's second method to the continuous-time best-response dynamics. Therefore, we need to show that there exists a radially unbounded Lyapunov function for the dynamics for which it holds that its time derivative is negative everywhere except at the pure Nash equilibrium $\boldsymbol{s}^{\ast}$. We consider the following candidate function:
    \begin{align*}
        V(\boldsymbol{s}^t)=\Phi(\boldsymbol{s}^{\ast})-\Phi(\boldsymbol{s}^t).
    \end{align*}
    Clearly, this is a valid Lyapunov function for the dynamics, since $V(\boldsymbol{s}^{\ast})=0$, and $V(\boldsymbol{s}^t)>0$ if $\boldsymbol{s}^t\neq\boldsymbol{s}^{\ast}$ where $\boldsymbol{s}^{\ast}$ is the unique maximiser of the potential function $\Phi$ (by Theorem \ref{theorem:uniqueness}). As it clearly holds that $V(\boldsymbol{x})\to\infty$ as $\|\boldsymbol{x}\|\to\infty$, the function $V$ is radially unbounded. In addition, let $\hat{s}_i^t$ be the unique strategy in $\text{BR}_{i}(\boldsymbol{s}^t)$, and then the time derivative of $V(\boldsymbol{s}^t)$ is given by
    \begin{align}\label{equation:time_derivative}
        \dfrac{dV(\boldsymbol{s}^t)}{dt}=\sum_{i\in L}\sum_{j\in B}\dfrac{\partial V(\boldsymbol{s}^t)}{\partial s_{ij}^t}\cdot\dfrac{ds_{ij}^t}{dt}=-\sum_{i\in L}\sum_{j\in B}\dfrac{\partial\Phi(\boldsymbol{s}^t)}{\partial s_{ij}^t}(\hat{s}_{ij}^t-s_{ij}^t). 
    \end{align}
    For $i\in L$, we consider the $i$-th term of the outer summation of (\ref{equation:time_derivative}),
    \begin{align}\label{equation:one_term}
        \sum_{j\in B}\dfrac{\partial\Phi(\boldsymbol{s}^t)}{\partial s_{ij}^t}(\hat{s}_{ij}^t-s_{ij}^t).
    \end{align}
    Then, define $\boldsymbol{v}$ by
    \begin{align*}
        v_{kj}=\left\{
        \begin{array}{ll}
            \hat{s}_{ij}^t-s_{ij}^t & k=i, \\
            0 & k\neq i, \\
        \end{array}
        \right. \quad \forall j\in B
    \end{align*}
    Equation (\ref{equation:one_term}) can then be written as $\boldsymbol{v}\cdot\nabla\Phi(\boldsymbol{s}^t)$. Now observe that the strategy profile $(\hat{s}_i^t,\boldsymbol{s}_{-i}^t)$ can be reached by moving from $\boldsymbol{s}^t$ in the direction $\boldsymbol{v}$, and $\Phi(\hat{s}_i^t,\boldsymbol{s}_{-i}^t)-\Phi(\boldsymbol{s}^t)=u_{i}(\hat{s}_i^t,\boldsymbol{s}_{-i}^t)-u_{i}(\boldsymbol{s^t})\geq0$, where the inequality is strict whenever $\hat{s}_i^t\neq s_i^t$. Concavity of $\Phi$ along the line segment between $\boldsymbol{s}$ and $(\hat{s}_i^t,\boldsymbol{s}_{-i}^t)$ then implies that $\boldsymbol{v}\cdot\nabla\Phi(\boldsymbol{s}^t)$ is non-negative, and is positive if $\hat{s}_i^t\neq s_i^t$. As this holds for every term of the outer summation of (\ref{equation:time_derivative}), we can conclude that $dV(\boldsymbol{s}^t)/dt<0$ if $s_i^t\notin\text{BR}_{i}(\boldsymbol{s}^t)$ for at least one $i\in L$. Noting that $\boldsymbol{s}^t\neq\boldsymbol{s}^{\ast}$ if there exists $k\in L$ such that $s_{k}^t\notin\text{BR}_k(\boldsymbol{s}^t)$, we find that for all $\boldsymbol{s}^t\in\boldsymbol{S}\setminus\{\boldsymbol{s}^{\ast}\}$,
    \begin{align*}
        \frac{dV(\boldsymbol{s}^t)}{dt}<0.
    \end{align*}
    Thus, according to Lyapunov's second method, the continuous-time best-response dynamics is globally asymptotically stable, which indicates that the dynamics will converge to the unique pure Nash equilibrium $\boldsymbol{s}^{\ast}$ from any initial strategy profile.
\end{proof}

\balance

\section{Conclusion}
In this paper, we define a non-cooperative lending game in the interbank money market in which lending banks strategically allocate their budget to borrowing banks. We show that this lending game is an exact potential game with a unique pure Nash equilibrium. We then characterise the pure Nash equilibrium of the game and propose a strongly polynomial-time algorithm to compute the equilibrium. From an economic perspective, our results indicate that even if the interbank money market is an over-the-counter market, the interest rates offered by borrowing banks are identical at equilibrium, which implies that information exchange and perfect competition among banks still guide the market towards a market interest rate. In addition, we study some variants of best-response dynamics of the lending game and prove their convergence to the pure Nash equilibrium in both discrete-time and continuous-time scenarios, indicating that banks will autonomously reach and remain stable at the pure Nash equilibrium under any initial market conditions. Some of our results on the convergence of best-response dynamics apply more generally to continuous games with strictly concave potential functions; therefore, our results on best-response dynamics may be of independent (and more general) interest.



\begin{acks}
    Bart de Keijzer was supported by the EPSRC grant EP/X021696/1. Carmine Ventre was supported by the Engineering and Physical Sciences Research Council (EPSRC) and UKFin+ network (Grant number EP/W034042/1). Jinyun Tong was supported by King's College London via the KDC and the China Scholarship Council via the K-CSC Scholarship.
\end{acks}



\bibliographystyle{ACM-Reference-Format} 
\bibliography{References}

\appendix

\section{Proof of Theorem 2.1}\label{apx:theorem:potential_game}
\begin{proof}
    Recall that the \textit{interest rate} offered by a borrower $j\in B$ is given by
    \begin{equation}\label{equation:interest_rate_repeated}
        r_j(\boldsymbol{s})=(r_\text{min}-r_\text{max})\frac{\sum_{i\in L}s_{ij}}{d_j}+r_\text{max}
    \end{equation}
    where $r_\text{min}$ and $r_\text{max}$ are the minimum and maximum interest rates, satisfying $0<r_\text{min}<r_\text{max}$. And for each lender $i\in L$, the \textit{utility function} is given by \begin{equation}\label{equation:utility_function_repeated}
        u_i(\boldsymbol{s})=\sum_{j\in B}(r_j(\boldsymbol{s})-r_\text{min})s_{ij}.
    \end{equation}
    An \textit{exact potential game} is a game for which there exists a \textit{potential function} $\Phi$, such that 
    \begin{equation}\label{equation:potential_function_proof_repeated}
         \Phi(s_k^\prime,\boldsymbol{s}_{-k})-\Phi(s_k,\boldsymbol{s}_{-k})=u_k(s_k^\prime,\boldsymbol{s}_{-k})-u_k(s_k,\boldsymbol{s}_{-k})
    \end{equation}
    holds for each player $k\in L$, any two strategies $s_k,s_k^\prime\in S_k$ and each strategy profile $\boldsymbol{s}_{-k}\in\boldsymbol{S}_{-k}$. For a subset $[z]=\{1,\ldots,z\}$ of the first $z$ lenders in $L$, let
    \begin{equation}\label{equation:interest_rate_first_z_repeated}
        r_j(\boldsymbol{s},z)=(r_\text{min}-r_\text{max})\frac{\sum_{i\in[z]}s_{ij}}{d_j}+r_\text{max}.
    \end{equation}
    Consider the following function as the candidate of the potential function of the lending game $G$.
    \begin{equation}\label{equation:potential_function_repeated}
        \Phi(\boldsymbol{s})=\sum_{j\in B}\sum_{i\in L}(r_j(\boldsymbol{s},i)-r_\text{min})s_{ij}.
    \end{equation}
    
    Now, we need to prove that (\ref{equation:potential_function_proof_repeated}) holds for any $k\in L$, $s_k,s_k^\prime\in S_k$, and $\boldsymbol{s}_{-k}\in\boldsymbol{S}_{-k}$. According to the definitions in (\ref{equation:interest_rate_first_z_repeated}) and (\ref{equation:potential_function_repeated}), the potential function can be written as
    \begin{equation*}
        \Phi(\boldsymbol{s})=\displaystyle\sum_{j\in B}\sum_{i\in L}\left((r_\text{min}-r_\text{max})\frac{\sum_{\ell\in [i]}s_{\ell j}}{d_j}+r_\text{max}-r_\text{min}\right)s_{ij}.
    \end{equation*}
    Let $k$ be an arbitrary bank switching its strategy from $s_k$ to $s_k^\prime=(s_{k1}^\prime,\dots,s_{kn}^\prime)$, while the other banks remain with their current strategies; therefore, the strategy profile shifts from $\boldsymbol{s}=(s_k,\boldsymbol{s}_{-k})$ to $\boldsymbol{s}^\prime=(s_k^\prime,\boldsymbol{s}_{-k})$. Then, the potential function after the strategy switch can be written as
    \begin{align*}
        & \Phi(s_k^\prime,\boldsymbol{s}_{-k}) \\ & \qquad =\displaystyle\sum_{j\in B}\left(\sum_{i\in L\setminus\{k\}}\left((r_\text{min}-r_\text{max})\dfrac{\sum_{\ell\in[i]\setminus\{k\}}s_{\ell j}+s_{kj}^\prime\boldsymbol{1}[k\leq i]}{d_j} \right.\right. \\
        & \qquad \qquad  +r_\text{max}-r_\text{min}\Bigg)s_{ij} \\
        & \qquad \left.+\left((r_\text{min}-r_\text{max})\dfrac{\sum_{\ell\in[k-1]}s_{\ell j}+s_{kj}^\prime}{d_j}+r_\text{max}-r_\text{min}\right)s_{kj}^\prime\right)
    \end{align*}
    where $\boldsymbol{1}[A]$ denotes the indicator function that maps to $1$ if $A$ is true, and to $0$ otherwise. Therefore, the change in potential is
    \begin{align*}
        & \Phi(s_k^\prime,\boldsymbol{s}_{-k})-\Phi(s_k,\boldsymbol{s}_{-k}) \\
        & \quad =\displaystyle\sum_{j\in B}\left(\sum_{i\in L\setminus\{k\}}\left((r_\text{min}-r_\text{max})\dfrac{\sum_{\ell\in[i]\setminus\{k\}}s_{\ell j}+s_{kj}^\prime\boldsymbol{1}[k\leq i]}{d_j} \right.\right. \\
        & \qquad \qquad +r_\text{max}-r_\text{min}\Bigg)s_{ij} \\
        & \qquad \left.+\left((r_\text{min}-r_\text{max})\dfrac{\sum_{\ell\in[k-1]}s_{\ell j}+s_{kj}^\prime}{d_j}+r_\text{max}-r_\text{min}\right)s_{kj}^\prime\right) \\
        & \qquad -\displaystyle\sum_{j\in B}\left(\sum_{i\in L\setminus\{k\}}\left((r_\text{min}-r_\text{max})\dfrac{\sum_{\ell\in[i]\setminus\{k\}}s_{\ell j}+s_{kj}\boldsymbol{1}[k\leq i]}{d_j} \right. \right. \\
        &\qquad \qquad +r_\text{max}-r_\text{min}\Bigg)s_{ij} \\
        & \qquad \left.+\left((r_\text{min}-r_\text{max})\dfrac{\sum_{\ell\in[k-1]}s_{\ell j}+s_{kj}}{d_j}+r_\text{max}-r_\text{min}\right)s_{kj}\right).
    \end{align*}
    For the first and third fractions, the innermost summations and the $r_\text{max}-r_\text{min}$ terms cancel out, so we can simplify and rewrite the expression as
    
    \begin{align*}
        & \Phi(s_k^\prime,\boldsymbol{s}_{-k})-\Phi(s_k,\boldsymbol{s}_{-k}) \\
        & \quad =\displaystyle\sum_{j\in B}\left(\sum_{i\in L\setminus\{k\}}(r_\text{min}-r_\text{max})\dfrac{(s_{kj}^\prime-s_{kj})\boldsymbol{1}[k\leq i]}{d_j}s_{ij}\right. \\
        & \qquad \left.+(r_\text{min}-r_\text{max})\dfrac{\sum_{\ell\in[k-1]}s_{\ell j}}{d_j}(s_{kj}^\prime-s_{kj})\right. \\
        & \qquad \left.+(r_\text{min}-r_\text{max})\dfrac{s_{kj}^\prime s_{kj}^\prime-s_{kj}s_{kj}}{d_j}+(r_\text{max}-r_\text{min})(s_{kj}^\prime-s_{kj})\right) \\
        & \quad =\displaystyle\sum_{j\in B}\left((r_\text{min}-r_\text{max})\dfrac{\sum_{i\in L\setminus[k]}s_{ij}}{d_j}(s_{kj}^\prime-s_{kj})\right. \\
        & \qquad \left.+(r_\text{min}-r_\text{max})\dfrac{\sum_{i\in[k-1]}s_{ij}}{d_j}(s_{kj}^\prime-s_{kj})\right. \\
        & \qquad \left.+(r_\text{min}-r_\text{max})\dfrac{s_{kj}^\prime s_{kj}^\prime-s_{kj}s_{kj}}{d_j}+(r_\text{max}-r_\text{min})(s_{kj}^\prime-s_{kj})\right) \\
        & \quad =\displaystyle\sum_{j\in B}\left((r_\text{min}-r_\text{max})\dfrac{\sum_{i\in L\setminus\{k\}}s_{ij}+s_{kj}^\prime}{d_j}+r_\text{max}-r_\text{min}\right)s_{kj}^\prime \\
        & \qquad -\displaystyle\sum_{j\in B}\left((r_\text{min}-r_\text{max})\dfrac{\sum_{i\in L\setminus\{k\}}s_{ij}+s_{kj}}{d_j}+r_\text{max}-r_\text{min}\right)s_{kj}.
    \end{align*}
    According to the definitions in (\ref{equation:interest_rate_repeated}) and (\ref{equation:utility_function_repeated}), the expression is equivalent to
    \begin{align*}
        & \Phi(s_k^\prime,\boldsymbol{s}_{-k})-\Phi(s_k,\boldsymbol{s}_{-k}) \\ & \qquad =\displaystyle\sum_{j\in B}\left(r_j(s_k^\prime,\boldsymbol{s}_{-k})-r_\text{min}\right)s_{kj}^\prime-\displaystyle\sum_{j\in B}(r_j(s_k,\boldsymbol{s}_{-k})-r_\text{min})s_{kj} \\
        & \qquad =u_k(s_k^\prime,\boldsymbol{s}_{-k})-u_k(s_k,\boldsymbol{s}_{-k}).
    \end{align*}
    Therefore, the game $G$ is an exact potential game with potential function $\Phi$.
\end{proof}

\section{Proof of Lemma 3.2}\label{apx:lemma:strict_concavity}
\begin{proof}
    Consider two arbitrary but different strategies $\boldsymbol{s}=(s_{ij})_{i\in L,j\in B}\in\boldsymbol{S}$ and $\boldsymbol{s}^\prime=(s_{ij}^\prime)_{i\in L,j\in B}\in\boldsymbol{S}$ such that $\boldsymbol{s}\neq\boldsymbol{s}^\prime$. Let $\boldsymbol{s}^\lambda=\lambda\boldsymbol{s}+(1-\lambda)\boldsymbol{s}^\prime$ where $\lambda\in(0,1)$. Then,
    \begin{align*}
        & \lambda\Phi(\boldsymbol{s})+(1-\lambda)\Phi(\boldsymbol{s}^\prime)-\Phi(\boldsymbol{s}^\lambda) \\
        & \quad =\lambda\displaystyle\sum_{j\in B}\sum_{i\in L}(r_j(\boldsymbol{s},i)-r_\text{min})s_{ij}+(1-\lambda)\sum_{j\in B}\sum_{i\in L}\left(r_j(\boldsymbol{s}^\prime,i)-r_\text{min}\right)s_{ij}^\prime \\
        & \qquad -\displaystyle\sum_{j\in B}\sum_{i\in L}\left(r_j(\lambda\boldsymbol{s}+(1-\lambda)\boldsymbol{s}^\prime,i)-r_\text{min}\right)(\lambda s_{ij}+(1-\lambda)s_{ij}^\prime) \\
        & \quad =\displaystyle\sum_{j\in B}\sum_{i\in L}\left[\left((r_\text{min}-r_\text{max})\dfrac{\sum_{\ell\in[i]}s_{\ell j}}{d_j}+r_\text{max}-r_\text{min}\right)\lambda s_{ij}\right. \\
        & \qquad \left.+\left((r_\text{min}-r_\text{max})\dfrac{\sum_{\ell\in[i]}s_{\ell j}^\prime}{d_j}+r_\text{max}-r_\text{min}\right)(1-\lambda)s_{ij}^\prime\right. \\
        & \qquad -\left((r_\text{min}-r_\text{max})\dfrac{\sum_{\ell\in[i]}(\lambda s_{\ell j}+(1-\lambda)s_{\ell j}^\prime)}{d_j}+r_\text{max}-r_\text{min}\right) \\
        & \qquad \qquad \cdot (\lambda s_{ij}+(1-\lambda)s_{ij}^\prime)\Bigg] \\
        & \quad =\displaystyle\sum_{j\in B}\sum_{i\in L}\left[\dfrac{r_\text{min}-r_\text{max}}{d_j}\left(\lambda s_{ij}\sum_{\ell\in[i]}s_{\ell j}+(1-\lambda)s_{ij}^\prime\sum_{\ell\in[i]}s_{\ell j}^\prime\right.\right. \\
        & \qquad -\lambda^2 s_{ij}\sum_{\ell\in[i]}s_{\ell j}-\lambda(1-\lambda)s_{ij}\displaystyle\sum_{\ell\in[i]}s_{\ell j}^\prime \\ 
        & \qquad \left.\left. -\lambda(1-\lambda)s_{ij}^\prime\sum_{\ell\in[i]}s_{\ell j}-(1-\lambda)^2 s_{ij}^\prime\sum_{\ell\in[i]}s_{\ell j}^\prime\right)\right] \\
        & \quad =\displaystyle\sum_{j\in B}\sum_{i\in L}\dfrac{r_\text{min}-r_\text{max}}{d_j}\lambda(1-\lambda)(s_{ij}-s_{ij}^\prime)\left(\sum_{\ell\in[i]}s_{\ell j}-\sum_{\ell\in[i]}s_{\ell j}^\prime\right) \\
        & \quad =\lambda(1-\lambda)\displaystyle\sum_{j\in B}\dfrac{r_\text{min}-r_\text{max}}{2d_j}\left(\sum_{i\in L}(s_{ij}-s_{ij}^\prime)^2+\left(\sum_{i\in L}(s_{ij}-s_{ij}^\prime)\right)^2\right) \\
        & \quad <0,
    \end{align*}
    where in the last step we used that $\boldsymbol{s}\neq\boldsymbol{s}^\prime$ (i.e., there exists $i\in L,j\in B$ such that $s_{ij}\neq s_{ij}^{\prime}$), and that $r_\text{min}<r_\text{max}$, so that all terms in the final summation are non-positive, and at least one is strictly negative. Therefore, $\lambda\Phi(\boldsymbol{s})+(1-\lambda)\Phi(\boldsymbol{s}^\prime)<\Phi(\boldsymbol{s}^\lambda)$ holds for any $\lambda\in(0,1)$, and the potential function $\Phi$ is strictly concave over $\boldsymbol{S}$.
\end{proof}

\section{Proof of Theorem 3.4}\label{apx:theorem:solution}
\begin{proof}
    Let $\bar{L}=[\bar{m}]$ where $\bar{m}$ is the least index in $[m]\cup\{0\}$ for which it holds that
    \begin{equation}\label{equation:bar_m_repeated}
        c_{\bar{m}+1}>\frac{1}{m-\bar{m}+1}\left(\sum_{j\in B}d_j-\sum_{\ell\in[\bar{m}]}c_{\ell}\right).
    \end{equation}
    Then, recall the KKT multipliers $\boldsymbol{\mu}$ obtained in Section 3.2, i.e.,
    \begin{equation}\label{equation:multipliers_repeated}
    \resizebox{\linewidth}{!}{$
        \mu_i=\left\{
        \begin{array}{ll}
            (r_\text{min}-r_\text{max})\left(\dfrac{c_i}{\sum_{k\in B}d_k}-\dfrac{1}{m-\bar{m}+1}\left(1-\dfrac{\sum_{\ell\in\bar{L}}c_\ell}{\sum_{k\in B}d_k}\right)\right) & \text{if } i\in\bar{L}, \\
            0 & \text{if } i\in L\setminus\bar{L}.
        \end{array}
        \right. \\
    $}
    \end{equation}
    \begin{equation*}
        \mu_{ij}=0 \quad \forall i\in L,j\in B.
    \end{equation*}
    And recall the equilibrium strategy profile $\boldsymbol{s}^\ast$ obtained in Section 3.2. That is, for any $i\in L,j\in B$,
    \begin{equation}\label{equation:solution_repeated}
        s_{ij}^\ast=\left\{
        \begin{array}{ll}
            \dfrac{c_i}{\sum_{k\in B}d_k}\cdot d_j & \text{if } i\in\bar{L}, \\
            \dfrac{1}{m-\bar{m}+1}\left(1-\dfrac{\sum_{\ell\in\bar{L}}c_\ell}{\sum_{k\in B}d_k}\right)\cdot d_j & \text{if } i\in L\setminus\bar{L}. \\
        \end{array}
        \right.
    \end{equation}
    Now, we show that $\boldsymbol{\mu}$ and $\boldsymbol{s}^\ast$ given by (\ref{equation:multipliers_repeated}) and (\ref{equation:solution_repeated}), respectively, satisfy all KKT conditions.
    \begin{enumerate}
        \item \textbf{Primal feasibility}: Primal feasibility holds because 
        \begin{enumerate}
            \item For $i\in\bar{L}$, it follows directly from (\ref{equation:solution_repeated}) that $\sum_{j\in B}s_{ij}^\ast=c_i$,
            \item For $i\in L\setminus{\bar{L}}$, it follows from (\ref{equation:solution_repeated}) that
            \begin{equation*}
                \sum_{j\in B}s_{ij}^\ast=\frac{1}{m-\bar{m}+1}\left(\sum_{j\in B}d_j-\sum_{\ell\in\bar{L}}c_{\ell}\right)<c_{\bar{m}+1}\leq c_i,
            \end{equation*} 
            where the strict inequalities hold by the definition of $\bar{m}$ (see Section 3.2), and the last inequality holds by the fact that the lenders $L$ are assumed to be ordered non-decreasingly in budget.
        \end{enumerate}
        \item \textbf{Stationarity}: Recall the Lagrangian function defined in Section 3.2:
        \begin{equation*}
            \mathcal{L}(\boldsymbol{s},\boldsymbol{\mu})=\Phi(\boldsymbol{s})+\displaystyle\sum_{i\in L}\mu_i\left(c_i-\sum_{j\in B}s_{ij}\right)+\displaystyle\sum_{i\in L}\sum_{j\in B}\mu_{ij}s_{ij}.
        \end{equation*}
        The partial derivative of $\mathcal{L}(\boldsymbol{s},\boldsymbol{\mu})$ is given by
        \begin{equation}\label{equation:lagrangian_function_partial_derivative}
            \frac{\partial\mathcal{L}}{\partial s_{ij}}(\boldsymbol{s})=(r_\text{min}-r_\text{max})\left(\frac{1}{d_j}\left(s_{ij}+\sum_{k\in L}s_{kj}\right)-1\right)-\mu_i+\mu_{ij} 
        \end{equation}
        for all $i\in L,j\in B$.
        Stationarity follows from verifying that $\frac{\partial\mathcal{L}}{\partial s_{ij}}(\boldsymbol{s}^\ast)=0$ for all $i\in L,j\in B$. First, we note that
        \begin{align}
            \sum_{i\in L}s_{ij}^\ast & =\sum_{i\in\bar{L}}s_{ij}^\ast+\sum_{i\in L\setminus\bar{L}}s_{ij}^\ast \notag \\
            & =d_j\left(\frac{\sum_{\ell\in\bar{L}}c_\ell}{\sum_{k\in B}d_k}+\frac{m-\bar{m}}{m-\bar{m}+1}\left(1-\frac{\sum_{\ell\in\bar{L}}c_\ell}{\sum_{k\in B}d_k}\right)\right) \notag \\
            & =\frac{d_j}{m-\bar{m}+1}\left(m-\bar{m}+\frac{\sum_{\ell\in\bar{L}}c_\ell}{\sum_{k\in B}d_k}\right). \label{equation:sum_s_ij}
        \end{align} 
        For $i\in\bar{L},j\in B$, inserting the above expression along with (\ref{equation:solution_repeated}) and (\ref{equation:multipliers_repeated}) into (\ref{equation:lagrangian_function_partial_derivative}), we obtain
        \begin{align*}
            & \frac{\partial\mathcal{L}}{\partial s_{ij}}(\boldsymbol{s}^\ast) \\
            & =(r_\text{min}-r_\text{max})\left(\frac{1}{d_j}\left(\frac{c_i}{\sum_{k\in B}d_k}d_j \right. \right. \\
            & \qquad \left. \left. + \frac{d_j}{m-\bar{m}+1}\left(m-\bar{m}+\frac{\sum_{\ell\in\bar{L}}c_{\ell}}{\sum_{k\in B}d_k}\right)\right)-1\right)\\
            & \quad -(r_\text{min}-r_\text{max})\left(\frac{c_i}{\sum_{k\in B}d_k}-\frac{1}{m-\bar{m}+1}\left(1-\frac{\sum_{\ell\in\bar{L}}c_{\ell}}{\sum_{k\in B}d_k}\right)\right) \\
            & = (r_\text{min}-r_\text{max})\left(\frac{c_i}{\sum_{k\in B}d_k} \right. \\ 
            & \qquad \left. +\frac{1}{m-\bar{m}+1}\left(m-\bar{m}+\frac{\sum_{\ell\in\bar{L}}c_{\ell}}{\sum_{k\in B}d_k}\right)-1\right) \\
            & \quad -(r_\text{min}-r_\text{max})\left(\frac{c_i}{\sum_{k\in B}d_k}-\frac{1}{m-\bar{m}+1}\left(1-\frac{\sum_{\ell\in\bar{L}}c_{\ell}}{\sum_{k\in B} d_k}\right)\right) \\
            & =0.
        \end{align*}
        For $i\in L\setminus\bar{L}$ and $j\in B$, we have $\mu_i=\mu_{ij}=0$, and then combine (\ref{equation:solution_repeated}) and (\ref{equation:sum_s_ij}), with (\ref{equation:lagrangian_function_partial_derivative}), 
        \begin{align*}
            \frac{\partial\mathcal{L}}{\partial s_{ij}}(\boldsymbol{s}^\ast) & =(r_\text{min}-r_\text{max})\left(\frac{1}{d_j}\left(\frac{1}{m-\bar{m}+1}\left(1-\frac{\sum_{\ell\in\bar{L}} c_{\ell}}{\sum_{k\in B}d_k}\right)d_j\right.\right. \\ 
            & \quad \left.\left.+\frac{d_j}{m-\bar{m}+1}\left(m-\bar{m}+\frac{\sum_{i\in\bar{L}}c_i}{\sum_{k\in B} d_k}\right)\right)-1\right) \\
            & = (r_\text{min}-r_\text{max})\left(\left(\frac{1}{m-\bar{m}+1}\left(1-\frac{\sum_{\ell\in\bar{L}} c_{\ell}}{\sum_{k\in B} d_k}\right)\right.\right. \\ 
            & \quad \left.\left.+\frac{1}{m-\bar{m}+1}\left(m-\bar{m}+\frac{\sum_{i\in\bar{L}}c_i}{\sum_{k\in B}d_k}\right)\right)-1\right) \\
            & =0.
        \end{align*}
        Then, We conclude that the stationarity conditions hold.
    
        \item \textbf{Dual feasibility}: The dual feasibility conditions hold trivially for $\mu_{ij}$ for all $i\in L,j\in B$, and also for $\mu_i$ in case $i\in L\setminus\bar{L}$.
        
        For $i\in\bar{L}$, observe that the first factor $(r_\text{min}-r_\text{max})$ in the expression for $\mu_i$ in (\ref{equation:multipliers_repeated}) is negative. We thus need to prove that the second factor is non-positive. To that end, we show firstly that $c_{\bar{m}}\leq\frac{1}{m-\bar{m}+1}\left(\sum_{i\in B}d_j-\sum_{\ell\in\bar{L}}c_\ell\right)$. Suppose for the sake of contradiction that the opposite holds: $c_{\bar{m}}>\frac{1}{m-\bar{m}+1}\left(\sum_{i\in B}d_j-\sum_{\ell\in\bar{L}}c_\ell\right)$. Multiplying both sides by $m-\bar{m}+1$ and adding $c_{\bar{m}}$ to both sides of the inequality yields $(m-\bar{m}+2)c_{\bar{m}}>\sum_{j\in B}d_j-\sum_{\ell\in[\bar{m}-1]}c_{\ell}$. However, it would mean that (\ref{equation:bar_m_repeated}) holds for an index less than $\bar{m}$ (namely, it holds for $\bar{m}-1$), and contradicts the definition of $\bar{m}$ as the minimum index for which (\ref{equation:bar_m_repeated}) holds. Therefore, it implies that
        \begin{equation*}
            c_i\leq c_{\bar{m}}\leq\frac{1}{m-\bar{m}+1}\left(\sum_{j\in B}d_j-\sum_{\ell\in\bar{L}}c_{\ell}\right) \qquad \forall i\in\bar{L}.
        \end{equation*}
        By dividing both sides by $\sum_{j\in B}d_j$, we obtain that 
        \begin{equation*}
            \frac{c_i}{\sum_{j\in B}d_j}-\frac{1}{m-\bar{m}+1}\left(1-\frac{\sum_{\ell\in\bar{L}} c_{\ell}}{\sum_{j\in B}d_j}\right)\leq0,
        \end{equation*}
        which indicates that the right-hand side of the second factor of (\ref{equation:multipliers_repeated}) is non-negative. Therefore, it holds that $\mu_i\geq0$ for all $i\in L$, and thus all dual feasibility conditions hold true.
    
        \item \textbf{Complementary slackness}: The fact that the complementary slackness conditions are satisfied follows immediately from observing that $\mu_{ij}=0$ for all $i\in L,j\in B$, that $\mu_i=0$ for $i\in L\setminus\bar{L}$, and that $c_i-\sum_{j\in B}s_{ij}^\ast=0$ for $i\in\bar{L}$.
    \end{enumerate}
\end{proof}

\section{Proof of Lemma 4.2}\label{apx:lemma:gradient_bound}
\begin{proof}
    We can bound $\boldsymbol{v}\cdot\nabla\Phi(\boldsymbol{s^\prime})$ as follows.
    \begin{align*}
        \boldsymbol{v}\cdot\nabla\Phi(\boldsymbol{s}^\prime) & =\sum_{i,j}v_{ij}(\nabla\Phi(\boldsymbol{s}^\prime))_{ij} \\
        & \geq\sum_{i,j}v_{ij}((\nabla\Phi(\boldsymbol{s}))_{ij}-a(s_{ij}^\prime-s_{ij})) \\ 
        & =\boldsymbol{v}\cdot\nabla\Phi(\boldsymbol{s})-\sum_{i,j}v_{ij}a(s_{ij}^\prime-s_{ij}) \\
        & \geq\boldsymbol{v}\cdot\nabla\Phi(\boldsymbol{s})-\max_{i,j}|v_{ij}|a\sum_{i,j}\left|s_{ij}^\prime-s_{ij}\right| \\
        &  \geq\boldsymbol{v}\cdot\nabla\Phi(\boldsymbol{s})-\max_{i,j}|v_{ij}|a\frac{\boldsymbol{v}\cdot\nabla\Phi(\boldsymbol{s})}{2a\max_{i,j}|v_{ij}|}=\frac{\boldsymbol{v}\cdot\nabla\Phi(\boldsymbol{s})}{2}.
    \end{align*}
\end{proof}

\section{Proof of Proposition 4.7}\label{apx:prop47}
\begin{proof}
    By computing second derivatives of the utility functions, we observe that the Jacobian $\nabla g(\boldsymbol{s})$ of $g(\boldsymbol{s})$ has the entries
    \begin{equation*}
        \frac{\partial^2u_i(\boldsymbol{s})}{\partial s_{ij}\partial s_{k\ell}}=\left\{
        \begin{array}{ll}
            0 & \text{if } k\neq i \text{ and } \ell\neq j, \\
            (r_\text{min}-r_\text{max})/d_j & \text{if } k\neq i \text{ and } \ell=j, \\
            2(r_\text{min}-r_\text{max})/d_j & \text{if } k=i \text{ and } \ell=j. \\
        \end{array}
        \right.
    \end{equation*}
    Note that $\nabla g(\boldsymbol{s})$ is constant for all $((i,j),(k,\ell))\in(L\times B)^2$. In fact, it is equal to the Hessian of $\Phi$ which we used in the proof of Theorem 4.1.
    
    Since $\nabla g(\boldsymbol{s})$ is symmetric, and negative definiteness is preserved under summation and transposition, $\nabla g(\boldsymbol{s})+(\nabla g(\boldsymbol{s}))^\top$ is negative definite if $\nabla g(\boldsymbol{s})$ is negative definite. By reordering rows and columns so that coordinates are grouped per borrower, rather than per lender, we obtain
    \begin{equation*}
        \boldsymbol{J}=\texttt{diag}(J_1,\dots,J_n)
    \end{equation*}
    where $\texttt{diag}$ denotes block-diagonal matrix composition and for $j\in B$, $J_j$ is an $m\times m$ matrix where
    \begin{equation*}
        (J_j)_{ik}=\frac{\partial^2u_i(\boldsymbol{s})}{\partial s_{ij}\partial s_{kj}}=\left\{
        \begin{array}{ll}
            (r_\text{min}-r_\text{max})/d_j & \text{if } k\neq i, \\
            2(r_\text{min}-r_\text{max})/d_j & \text{if } k=i. \\
        \end{array}
        \right.
    \end{equation*}
    Considering that $r_\text{min}<r_\text{max}$, $\boldsymbol{J}$ is negative definite. In turn, since negative definiteness is preserved under block-diagonal matrix composition, it suffices to show that for every $j\in B$, $J_j$ is negative definite. To verify this, we simply observe, by using the expressions we established on the entries of $J_j$, that for any non-zero vector $\boldsymbol{v}\in\mathbb{R}^m$,
    \begin{equation*}
        \boldsymbol{v}^\top J_j\boldsymbol{v}=\frac{r_\text{min}-r_\text{max}}{d_j}\left(\sum_{i\in L}v_i^2+\left(\sum_{i\in L}v_i\right)^2\right)<0.
    \end{equation*}
    Therefore, we prove that $\nabla g(\boldsymbol{s})+(\nabla g(\boldsymbol{s}))^\top$ is negative definite.
\end{proof}

\end{document}